\documentclass[11pt]{article}

\usepackage{amsmath,amsfonts,amsthm,amssymb,color}
\usepackage{fancybox}
\usepackage{tikz}
\usepackage{hyperref}
\usepackage{cleveref} % must come after hyperref
\usepackage{thm-restate}
\crefformat{equation}{(#2#1#3)}
\usepackage{algorithm}
\usepackage{algorithmicx, algpseudocode}

\usepackage{caption}
\usepackage{bbm}
\usepackage{tablefootnote}
\usepackage{bbm}

\usepackage[margin=1.0in]{geometry}

\newtheorem{theorem}{Theorem}[section]
\newtheorem{corollary}[theorem]{Corollary}
\newtheorem{lemma}[theorem]{Lemma}

\newtheorem{proposition}[theorem]{Proposition}

\newtheorem{fact}[theorem]{Fact}

\theoremstyle{definition}
\newtheorem{definition}[theorem]{Definition}
\newtheorem{remark}[theorem]{Remark}

\newenvironment{fminipage}%
{\begin{Sbox}\begin{minipage}}%
		{\end{minipage}\end{Sbox}\fbox{\TheSbox}}

\def\var#1{\mbox{\bf Var}\left[ #1 \right]}

\def\defeq{\stackrel{\mathrm{def}}{=}}

\def\abs#1{\left|#1  \right|}

\def\norm#1{\left\| #1 \right\|}

\renewcommand\Pr{\operatorname{Pr}}

\newcommand{\mD}{\mathcal{D}}
\newcommand{\bR}{\mathbb{R}}
\newcommand{\cD}{\mathcal{D}}
\newcommand{\cM}{\mathcal{M}}
\newcommand{\cR}{\mathcal{R}}

\DeclareMathOperator{\Lap}{Lap}

\newcommand{\fitting}{{\sc Preprocessing}}

\newcommand\eps{\epsilon}

\newcommand{\upper}{\textsc{Upper}}
\newcommand{\lowerbound}{\textsc{Lower}}
\newcommand{\newSens}{{\Delta}}

\newcommand{\hmu}{\hat{\mu}}
\newcommand{\R}{\mathbb{R}}
\newcommand{\sensitivity}{Sensitivity-Preprocessing Function}

\begin{document}
	
\title{Individual Sensitivity Preprocessing for Data Privacy}

\author{
Rachel Cummings\footnote{School of Industrial and Systems Engineering, Georgia Institute of Technology.  Email: \texttt{rachelc@gatech.edu}. Supported in part by a Mozilla Research Grant.}
\\
Georgia Tech
\and
David Durfee\footnote{School of Computer Science, Georgia Institute of Technology.  Email: \texttt{ddurfee@gatech.edu}. This material is based upon work supported by the National Science Foundation under Grants No. 1718533.}
\\
Georgia Tech
}

	\maketitle
	\begin{abstract}
		The sensitivity metric in differential privacy, which is informally defined as the largest marginal change in output between neighboring databases, is of substantial significance in determining the accuracy of private data analyses.  Techniques for improving accuracy when the average sensitivity is much smaller than the worst-case sensitivity have been developed within the differential privacy literature, including tools such as smooth sensitivity, Sample-and-Aggregate, Propose-Test-Release, and Lipschitz extensions.

In this work, we provide a new and general Sensitivity-Preprocessing framework for reducing sensitivity, where efficient application gives state-of-the-art accuracy for privately outputting the important statistical metrics median and mean when no underlying assumptions are made about the database. In particular, our framework compares favorably to smooth sensitivity for privately outputting median, in terms of both running time and accuracy. Furthermore, because our framework is a preprocessing step, it can also be complementary to smooth sensitivity and any other private mechanism, where applying both can achieve further gains in accuracy.

We additionally introduce a new notion of \textit{individual sensitivity} and show that it is an important metric in the variant definition of \textit{personalized differential privacy}. We show that our algorithm can extend to this context and serve as a useful tool for this variant definition and its applications in markets for privacy.

Given the effectiveness of our framework in these important statistical metrics, we further investigate its properties and show that:
(1) Our construction is conducive to efficient implementation with strong accuracy guarantees, evidenced by an $O(n)$ implementation for median (with presorted data), and $O(n^2)$ implementation for more complicated functions such as mean, $\alpha$-trimmed mean, and variance.
(2) Our construction is both NP-hard and also optimal in the general setting
(3) Our construction can be extended to higher dimensions, although it incurs accuracy loss that is linear in the dimension.

	\end{abstract}
	\thispagestyle{empty}
	\newpage
\pagenumbering{arabic}

% !TEX root = main.tex

\section{Introduction}\label{sec:intro}

Differentially private algorithms for data analysis guarantee that any individual entry in a database has only a bounded effect on the outcome of the analysis \cite{DMNS06}.  These algorithms ensure that the outcomes on any pair of neighboring databases---that differ in a single entry---are nearly indistinguishable.  This is typically achieved by perturbing the analysis or its output, using noise that scales with the magnitude of change in the analysis between neighboring databases.  This perturbation necessarily leads to decreased accuracy of the analysis.  A fundamental challenge in differentially private algorithm design is to simultaneously satisfy privacy guarantees and provide accurate analysis of the database.  Privacy alone can be achieved by outputting pure noise, but this fails to yield useful insights about the data. Intuitively, stronger privacy guarantees should yield weaker accuracy guarantees.  Quantifying this privacy-accuracy tradeoff has been one major contribution of the existing differential privacy literature.  In the last several years, accurate and differentially private algorithms have been designed for a diverse collection of data analysis tasks (see \cite{DR14} for a survey), and have been implemented in practice by major organizations such as Apple, Google, Uber, and the U.S. Census Bureau.  The formal guarantees of differential privacy give sharp contrast to ad hoc privacy measures such as anonymization and aggregation, which have both led to infamous privacy violations \cite{NS08, HSR+08}.

%The \emph{global sensitivity} of a function is the maximum change in the function's value between all pairs of neighboring databases.  

We formalize data analysis tasks as functions that map from the space of all databases to real-valued outputs.  The \emph{global sensitivity} of a function is the worst-case difference in the function's value between all pairs of neighboring databases.  Since differential privacy guarantees must hold for all pairs of neighboring databases, this is the scale of noise that must be added to preserve privacy.  Strong bounds on global sensitivity imply that the function is well-behaved over the entire data universe, and often allows for privacy-preserving output with strong accuracy guarantees.  However, this worst-case measure allows a single outlier database to significantly skew the accuracy of the privacy-preserving algorithm for all databases.  Although it is necessary to preserve the privacy of outlying databases, we would prefer to add less noise for improved accuracy guarantees when the average-case sensitivity is far smaller than the worst-case.
A variety of well-known techniques have been employed to address this problem including smooth sensitivity and Sample-and-Aggregate \cite{NRS07}, Propose-Test-Release \cite{DL09}, and Lipschitz extensions \cite{KKRS13,BBDO13,RS16}.

%\david{previous sentence is trying to say "this problem is important, and previous work on it has been in stoc/focs"}

%Ideally, we would like to add noise that scales with the local sensitivity of each database, thus adding less noise on well-behaved regions of the database universe, and only adding substantial noise on the outliers.

%\david{we probably want to shorten the smooth sensitivity section}

Initial work in this space considered a database-specific definition of sensitivity, known as \emph{local sensitivity}, which is the maximum change in the function's value between a given database and its neighbors \cite{NRS07,DL09}.  Ideally, we would like to add noise that scales with the local sensitivity of each database.  This would allow us to add less noise to well-behaved regions of the database universe, and only the outliers would require substantial noise.  Unfortunately this procedure does not satisfy differential privacy because the amount of noise added to a given database may be highly disclosive.  To avoid this information leakage, \cite{NRS07} defined an intermediate notion of \emph{smooth sensitivity}, which smoothed the amount of noise added across databases to preserve differential privacy once again. This technique was also combined with random subset sampling to give an efficient and private procedure, Sample-and-Aggregate, with strong error guarantees when each database was well-approximated by a random subset of its entries \cite{NRS07}.

Later work considered partitioning the database universe into well-behaved and outlying databases. 
Propose-Test-Release defined this partition with respect to local sensitivity and gave accurate outputs only on databases that were sufficiently far from outliers \cite{DL09}.
Propose-Test-Release avoided some of the information leakage issues by outputting $\textsc{Null}$ for any outlying database, and gave efficient implementations for a variety of important functions.  The Lipschitz extension framework instead partitioned according to global sensitivity, by identifying a subset of the data universe where the given function had small global sensitivity.  On this subset, the function of interest is simply a Lipschitz function with the constant defined as the small global sensitivity \cite{KKRS13,BBDO13,RS16}.  Extending the Lipschitz function to the remaining data universe achieves a function with small global sensitivity that is identical to the original function on the well-behaved databases.
Applying any differential privacy algorithm to this Lipschitz function will allow for the use of a much smaller global sensitivity input and will achieve high accuracy on the well-behaved databases.

%In this way, the Lipschitz extensions can be thought of as a preprocessing step, and 
%Unfortunately, extending the Lipschitz function can be inefficient, and even uncomputable, significantly limiting the set of functions under which this construction is possible.

% \rc{david add sentence on lipschitz privacy}
 
%\textbf{Draw comparisons between ours and the previous results}

%\rc{introduce our results?}

In this work, we introduce a Sensitivity-Preprocessing framework that will similarly approximate a given function with a sensitivity bounded function, which we call the Sensitivity-Preprocessing Function. 
%While these techniques work well in a variety of settings, they still have limitations, some of which our framework will be able to overcome.
%The Sensitivity-Preprocessing framework introduced in this work will be most similar to Lipschitz extensions because it is a preprocessing step that constructs a sensitivity-bounded function, which we refer to as the \sensitivity, that approximates the original function.
Our Sensitivity-Preprocessing Function will take advantage of the specific metric space structure of the data universe to give a more constructive approach.
At a high-level, while Lipschitz extensions are initialized with a well-behaved subset of the data universe, our algorithm will find this well-behaved subset as it constructs the \sensitivity.
As a result, our procedure will be much more localized and can always give an exponential-time construction even in the most general setting, whereas Lipschitz extensions can often be uncomputable.
In addition, similar to smooth sensitivity, we achieve optimality and NP-hardness guarantees for accuracy in this generalized setting under several reasonable metrics of optimality.

Furthermore, our \sensitivity~only requires a simple recursive construction that is more conducive to efficient implementation, which we achieve for important statistical functions such as median, mean, and variance. 
%For these functions, we allow the inputs to these functions to be unbounded, which implies unbounded local sensitivity for all databases, since it is a worst-case metric over an infinite set.  Our construction takes advantage of further localization properties to handle this unbounded sensitivity setting.  
These functions have been of particular interest for similar techniques because they are highly important statistics and also have large worst-case sensitivity, but small average sensitivity.
We will compare our results to previous results in the following section, where our framework gives state-of-the-art accuracy for median and mean when no underlying assumptions are made to the database.
The key assumption that we would aim to avoid is that data points are drawn iid, which is a popular assumption in previous results for outputting functions such as $\alpha$-trimmed mean (i.e., Propose-Test-Release \cite{DL09}).
While this assumption is quite standard, we contend that real world data are often not iid, and so consideration of the more general case is still an important problem.
Comparing our framework to those that apply the iid assumptions (and sometimes further assumptions, such as being drawn from a Gaussian \cite{KLSU18}), we will also achieve similarly high accuracy for well concentrated databases, but concede that mechanisms specifically catered for that setting will often be superior.
However, we note that because our framework is a preprocessing routine, it can be run before applying any differentially private mechanism such as Propose-Test-Release to achieve further gains in accuracy.
In avoiding this iid assumption, the primary technique we will then compare our framework with will be smooth sensitivity which is popularly used for privately outputting median.
Both frameworks have database-specific accuracy and direct comparison will be difficult, but we give strong evidence in the next section of why our framework compares favorably to smooth sensitivity for median. %\david{Edited this paragraph. Needs to be cleaned}

The localized construction of our \sensitivity~will also allow us to tailor the new sensitivity parameters beyond previous techniques.
To this end, we introduce a more refined sensitivity metric, which we call \emph{individual sensitivity}, and show that it is important for a variant definition of \emph{personalized differential privacy} introduced in \cite{GR13}, and used in subsequent works on market design for private data \cite{CCKMV13, CLR+15}.  We can apply our construction as a preprocessing step for more refined sensitivity tailoring to take advantage of personalized differential privacy guarantees.  We believe this application of our results may be of independent interest for future work in these directions.

%Due to the preprocessing nature of our construction, we can apply our tool with the more refined sensitivity tailoring to take advantage of this variant definition, and we believe that it can be very useful for future work in this direction.

In this work we cover a broad range of the more immediate results from this new framework, but believe that there is still a substantial amount of work in this direction.
While some of our proofs will become involved, all of our results follow from first principles, suggesting the potential for further results using more sophisticated tools within this framework.
These further results include:
efficient implementations of more difficult functions such as linear regression; optimizing the trade-off between decreasing the sensitivity parameter and the error incurred by our Sensitivity-Preprocessing for specific functions; applying our algorithm in the markets for privacy literature; and variants of our algorithm that are optimized for specific computational settings or application domains.

\subsection{Our Results}

Our results will primarily revolve around the \textit{Sensitivity-Preprocessing Function}, which we introduce below.  It is an alternate schema for fitting a general function to a sensitivity-bounded function in the context of differential privacy.
More specifically, we consider the general problem of taking any function $f: \mD \to \R$ and constructing a new function $g: \mD \to \R$ that satisfies given sensitivity parameters and minimizes the difference $\abs{f(D) - g(D)}$ over all databases $D \in \mD$.  
The sensitivity parameters we consider will be more refined and we define \emph{individual sensitivity}, which is the maximum change in a function's value from adding or removing a single specific data entry.  We use $\Delta_i$ to denote individual sensitivity to the $i$-th data entry.  
When a database is comprised of data from multiple individuals, $\Delta_i$ captures the sensitivity of the function to person $i$'s data.

%To achieve differential privacy, it is sufficient to add noise that scales with the \emph{global sensitivity} of the function, which is the maximum change in the function's value from changing any single data entry.  In this work, we define \emph{individual sensitivity}, which is the maximum change in the function from adding or removing a specific data entry.  We will use $\Delta_i$ to denote the individual sensitivity of a function to the $i$-th entry of its input database.  When a database is comprised of data from multiple individuals, $\Delta_i$ denotes the sensitivity of the function to person $i$'s data.  For a function that treats all data entries symmetrically, all individual sensitivities equals the global sensitivity.

%As such, we will leverage the simplistic nature of this recursion to provide a variety of theoretical results and practical applications.
In this section, we first give an overview of our recursively constructed \sensitivity~that works in a highly generalized setting, along with the corresponding runtime and error guarantees.
We then examine the optimality and hardness of this general function in the context of minimizing $\abs{f(D) - g(D)}$ over all databases $D \in \mD$.
While constructing our \sensitivity~will require exponential time in general, we show that it can be simply and efficiently implemented in $O(n)$ time for median (with presorted data), and in $O(n^2)$ time for several other important statistical measures including mean and variance.  We show that our \sensitivity~tailors an important metric (individual sensitivity) in the variant definition of \textit{personalized differential privacy}, which provides different privacy guarantees to different individuals in the same database, and is a useful tool in the design of markets for privacy.  We further generalize our construction of \sensitivity~to bound the $\ell_1$ sensitivity of 2-dimensional functions $f: \mD \to \R^2$, and show that such techniques cannot be extended to higher dimensions without treating each dimension independently.

\subsubsection*{Sensitivity-preprocessing function overview}

%The primary goal in this result is to give an alternate schema for fitting a general function to a sensitivity bounded function.
%More specifically, we suppose we are given a function $f: \mD \to \bR$ and sensitivity parameters $\{\Delta_i\}$, and want to produce another function $g: \mD \to \bR$ that closely approximates $f$ on nearly all databases, and has individual sensitivity at most $\Delta_i$ for all $i$.

%Our construction of the \sensitivity~can be thought of similarly to Lipschitz extensions, but the set we extend from only contains one element, the empty set.
%It would then be easy to construct $g: \mD \to \bR$ with appropriate individual sensitivity bounds if we only require $f(\emptyset) = g(\emptyset)$. \rc{i don't get this sentence}
%Instead, we start with $f(\emptyset) = g(\emptyset)$, and will inductively construct $g$ for larger databases while trying achieve two desiderata: (1) maintain the appropriate individual sensitivity bounds; and (2) keep $g$ as close as possible to $f$. 
%The first objective will be strictly maintained and we will attempt to optimize over the second objective.

Our construction of the \sensitivity~is similar to the Lipschitz extension framework, but we extend only from the empty set.  We start with $f(\emptyset) = g(\emptyset)$, and inductively construct $g$ for larger databases while trying to achieve two desiderata: (1) maintain the appropriate individual sensitivity bounds; and (2) keep $g$ as close as possible to $f$. 
The first objective will be strictly maintained, and we will optimize over the second objective.

The primary difficulty in this construction is that we often consider $\mD$ to be infinite.
As a result, checking to make sure we do not violate any sensitivity constraints when defining $g$ on a new database can require checking all databases on which $g$ was previously defined.
For example, general Lipschitz extensions require checking all previously defined databases to extend to another database, which can often be uncomputable for general functions. 
To avoid these uncomputability issues, we take advantage of the lattice structure of neighboring databases in the differential privacy landscape.
This will allow us to give a far more localized construction that critically utilizes the following two key properties of the data universe metric space:
\begin{enumerate}
	\item While each database could have infinitely many neighboring databases, it only has a finite number of neighbors with strictly fewer entries.
	\item Any two neighbors of a strictly larger database must also be neighbors of a strictly smaller database.
\end{enumerate}

These properties ensure that whenever we define $g$ on a new database, we only need to check that sensitivity constraints of strictly smaller neighboring databases are satisfied.  Once we have found the feasible range of $g$ that does not violate any sensitivity constraints, we will define $g$ to be as close as possible to $f$ within this feasible range.  %Definition \ref{def.informal} gives an informal definition of our \sensitivity, given formally in Definition \ref{def:1D_sensitivity_function}.

\begin{definition}[Informal version of Definition \ref{def:1D_sensitivity_function}]\label{def.informal}
Given a function $f: \mathcal{D} \rightarrow \R$ and fixed sensitivity parameters, we recursively define our  \sensitivity~$g: \mathcal{D} \rightarrow \R$ such that $g(\emptyset) = f(\emptyset)$\footnote{See Remark \ref{rem.empty} for a discussion of how to initialize $g(\emptyset)$ if $f(\emptyset)$ is not well-defined.} and for any $D \in \mD$,
	\[
	g(D)= \text{closest point to $f(D)$ in }  \textsc{Feasible}(D),
	\]
where $\textsc{Feasible}(D)$ is the set of all points that do not violate the sensitivity constraints based upon $g(D')$ for all neighbors $D'$ of $D$ with fewer entries. 
%Given that this function will be recursively defined by looking at smaller neighbors, we start with the base case $g(\emptyset) = f(\emptyset)$.	
\end{definition}

The recursive structure of this function allows us to compute $g(D)$ by only looking at the subsets of $D$, which unfortunately takes exponential time.  Later in the paper (Sections \ref{sec:efficient_examples} and \ref{sec:variance_algo}), we utilize the simplicity of the recursive structure to efficiently implement this algorithm for several functions of interest that exhibit additional structure.  Theorem \ref{thm.inmain} summarizes our main result on the running time and accuracy guarantees of our \sensitivity~algorithm.

%Furthermore, the simplicity of the recursive structure will allow us to efficiently implement this algorithm for specific functions.
%Altogether, this gives us the following theorem, which we formalize in Section~\ref{sec:mainAlgo}, that requires minimal assumptions.

\begin{theorem}[Informal version of Theorem \ref{thm:main_1d}]\label{thm.inmain}
For any function $f: \mD \to \R$ and desired sensitivity bounds $\{\Delta_i\}$, let $g: \mD \to \R$ be the \sensitivity~of $f$.  Given query access to $f$ in $T(n)$ time for a database of size $n$, we give $O((T(n) + n)2^n)$ time access to $g(D)$ for any $D \in \mD$ with $n$ entries.  We also give instance-specific bounds on each $\abs{f(D) - g(D)}$ based on the sensitivity of $f$ and $\{\Delta_i\}$.
\end{theorem}

Our algorithm for computing the \sensitivity~$g$ (Algorithm \ref{algo.fitting}) is robust to informational assumptions.  We only assume query access to $f$, and do not require any knowledge of the database universe $\cD$ or the sensitivity of $f$.

This algorithm easily extends to functions that map to $\R^d$, by treating each dimension independently.  See Remark~\ref{rem:higher_dim} and Section \ref{sec:higher_dimensions} for more details on handling high-dimensional functions.

\subsubsection*{Approximate Optimality and Hardness}

The construction of our \sensitivity~is quite simple in its greedy structure and requires exponential running time. To justify these two properties, we complement our algorithm with both optimality guarantees and hardness results.

In particular, we still consider the general problem of taking any function $f: \mD \to \R$ and constructing a new function $g: \mD \to \R$ with individual sensitivity bounds $\{\Delta_i\}$.
The goal will then be to minimize the difference $\abs{f(D) - g(D)}$ across databases $D \in \mD$.  Despite its simplicity, we show that our \sensitivity~still achieves a 2-approximation to the optimal function in the $\ell_\infty$ metric in this generalized setting.

\begin{proposition}[Informal version of Corollary \ref{cor:2-approx_1d}]\label{prop.in2apx}
	Given any function $f: \mD \to \R$ and sensitivity parameters $\{\Delta_i\}$, let $g: \mD \to \R$ be our \sensitivity. For any function $f^*: \mD \to \R$ with individual sensitivity bounds $\{\Delta_i\}$,	
	\[
	\max_{D \in \mD} \abs{f(D) - g(D)} \leq 2 \max_{D \in \mD} \abs{f(D) - f^*(D)}.
	\]	
\end{proposition}

Our guarantees are even stronger because they also hold over finite subsets of the data universe.  While Proposition \ref{prop.in2apx} measures error in the worst-case over $\cD$, we also show (Lemma \ref{lem:2-approx_1d_key_lemma}) that when the optimal error is small on certain subsets of the data universe, then our error is also small.

Furthermore, we can show that our \sensitivity~is Pareto optimal: there is no strictly superior sensitivity-bounded function that improves accuracy over all databases.

\begin{proposition}[Informal version of Lemma \ref{lem:optimality}]
	Given any $f:\mathcal{D} \rightarrow \mathbb{R}$, let $g:\mathcal{D} \rightarrow \mathbb{R}$ be the \sensitivity~of $f$ with individual sensitivity parameters $\{ \Delta_i \}$. For any $f^*: \mD \to \R$ with individual sensitivity parameters $\{\Delta_i\}$, if there is some $D \in \mathcal{D}$ such that
	\[ \abs{f(D) - f^*(D)} < \abs{f(D) - g(D)},\]
	then there also exists some $D' \in \mathcal{D}$ such that 
	\[ \abs{f(D') - f^*(D')} > \abs{f(D') - g(D')}.\]
\end{proposition}

These results imply that our \sensitivity~does quite well fitting to the original function under the metrics we are considering.
However, it does take exponential time, so we complement these results by showing that getting the same approximation guarantees is NP-hard even for a single sensitivity parameter $\{\Delta_i\} = \Delta$.

\begin{proposition}[Informal version of Proposition \ref{prop:NP-hard}]
Given any function $f:\mathcal{D} \rightarrow \mathbb{R}$ and sensitivity parameter $\Delta$,	it is NP-hard to construct any function $f^*:\mD \to \R$ with sensitivity $\Delta$ that enjoys the same accuracy guarantees as our \sensitivity.
\end{proposition}

We further argue that it is uncomputable to do better than a 2-approximation in the $\ell_\infty$ metric, and also uncomputable to achieve even a constant approximation in any $\ell_p$ metric for $p < \infty$, which justifies our choice of metric. 
%In Section \ref{subsec:2-approximation_1d} we show that Algorithm \ref{algo.fitting} does indeed find a $g$ that 2-approximates $f$.
We believe that the combination of these results gives a strong indication that our \sensitivity~and corresponding exponential time construction is the best we can hope to achieve for the general problem under reasonable metrics of optimality.

\subsubsection*{Efficient Implementation for Important Statistical Measures}

One of the main benefits of our \sensitivity~is that its simple recursive structure is conducive to giving simple efficient variants for specific functions through largely straightforward state space reductions and dynamic programming.
To this end, we give efficient implementations of our \sensitivity~for the important statistical functions mean, median, $\alpha$-trimmed mean, maximum, minimum, and variance.
These statistical metrics can be surprisingly difficult to release privately without assuming the input is restricted to some range, and often requires further assumptions for metrics like mean, such as data being drawn from identical and independent distributions.
In fact, for Propose-Test-Release even the iid assumption is not sufficient to apply their framework to mean, and other works required further concentration properties such as being drawn from the normal distribution.
However, our \sensitivity~does not require bounded sensitivity of the input function $f$, and can also avoid any iid requirements.
As a result, we are able to efficiently implement each of these statistical metrics with no constraints on the inputs.
It is important to note that our implementations only consider a single sensitivity parameter $\Delta$, but we believe each can be efficiently extended for individual sensitivity parameters $\{\Delta_i\}$. 

For each of these statistical metrics, we are simply implementing our \sensitivity~more efficiently, so all of the previously stated optimality guarantees still apply.
To further strengthen these optimality guarantees, we give a more rigorous treatment of the error incurred by our efficient implementation of median, mean, and variance, which we consider to be three of the most fundamental statistical tasks.

\paragraph{Median}  We focus on privately and accurately computing median, because it has been extensively studied under smooth sensitivity \cite{NRS07}.  Both our framework and smooth sensitivity provide database-specific accuracy guarantees, so a direct comparison of accuracy will be difficult.  Nevertheless, we show that our framework compares favorably to smooth sensitivity on median.

We begin by stating our result which will use a definition from \cite{NRS07} to give a more apparent comparison in terms of accuracy. As in \cite{NRS07}, we formally define 
\[
A^{(k)}(D) = \max_{d(D,D') \leq k} LS_f(D')
\]
as the $k$-local sensitivity of a function $f$ for database $D$. For median this just reduces to $A^{(k)}(D) = \max_{0 \leq t \leq k+1} (x_{m + t} - x_{m + t - k - 1})$ where $m = \frac{n+1}{2}$ (when $n$ is assumed to be odd). When $n$ is even, we define the median to be the average of the middle two entries (i.e. $(x_{\lfloor \frac{n + 1}{2} \rfloor} + x_{\lceil \frac{n + 1}{2} \rceil} ) / 2$) as is standard, and the resulting interpretation of $A^{(k)}(D)$ is nearly identical.

We also need to define median on the empty set.  Since this is not naturally defined, we allow it to be an input parameter $med(\emptyset)$ chosen by the data analyst as the estimated median. 
As our comparison will mostly be with smooth sensitivity which must assume values are in a bounded range $[min,max]$, the natural choice would be $med(\emptyset) = \frac{max - min}{2}$.
Further, it would be natural in this setting to set our parameter $\Delta = \frac{max - min}{n}$.
%The analyst's choice of $med(\emptyset)$ should reflect her prior knowledge, and will play a role in our accuracy guarantees.  %Our \sensitivity~will be perfectly accurate for databases with bounded $k$-local sensitivity and median close to $med(\emptyset)$.

\begin{restatable}{theorem}{median}
	\label{thm:median}
	Let $med: \R^{<\mathbb{N}} \rightarrow \R$ be the median function for the data universe of all finite-length real-valued vectors. For chosen parameters $med(\emptyset)$ and $\Delta$, along with any database $D = (x_1,....,x_n) \in \R^{<\mathbb{N}}$, if $x_1 \leq \cdots \leq x_n$ we give $O(n)$ time access to a function $g: \R^{<\mathbb{N}} \rightarrow \R$ with sensitivity $\Delta$ such that $g(D) = med(D)$ whenever $A^{(k)}(D) \leq 2(k + 1)\Delta$ for $k \leq n/4$ and $med(D) \in [med(\emptyset) - \frac{n}{2}\Delta, med(\emptyset) + \frac{n}{2}\Delta]$.
\end{restatable}

To interpret this accuracy result, we begin by comparing performance on the specific example that considered by \cite{NRS07}, where smooth sensitivity performed well.  
In fact, it was exactly this setting that motivated our choice of assumptions under which to show our framework outperforms smooth sensitivity.
Consider an environment where data points $x_1, \ldots, x_n$ lie in a bounded range $[0,1]$, and we naturally set $med(\emptyset) = 1/2$ and $\Delta = 1/n$.  Consider the particular database $D=(x_1, \ldots, x_n)$, where $x_i = i/n$.  In this example, it is easy to check that the assumptions are satisfied for our \sensitivity~to correctly output $g(D)=med(D)$.  Privately answering the query $g(D)$ using the Laplace mechanism (see Definition \ref{def.laplacemech} for a formal definition) outputs $med(D)$ plus noise with scale $\Delta / \epsilon = \frac{1}{\epsilon n}$.  In contrast the noise parameter added under smooth sensitivity (without our sensitivity preprocessing) would have to scale $\frac{1}{\epsilon^2 n}$, which is asymptotically larger, and would thus yield significantly lower accuracy.

The assumptions in Theorem \ref{thm:median} that $A^{(k)}(D) \leq 2(k + 1)\Delta$ for all $k \leq n/4$ and $med(D) \in [med(\emptyset) - \frac{n}{2}\Delta, med(\emptyset) + \frac{n}{2}\Delta]$ are then exactly the generalization of this condition where our database still has values that are reasonably spread out, but also has low local sensitivity and we would still like to achieve high accuracy.
Further note that these assumptions allow both the bottom and top quartile values to be arbitrarily small and large, implying that our construction is able to handle outliers well.
Restricting our attention to databases that satisfy these conditions allows us to consider all of the databases under which smooth sensitivity performs well. Under these assumptions, smooth sensitivity will achieve a noise magnitude of $\frac{\Delta}{\epsilon^2}$, whereas we instead achieve an asymptotically better noise magnitude of $\frac{\Delta}{\epsilon}$. 
Note that smooth sensitivity requires a bounded range, which we are considering here to be of size $n \Delta$, giving a global sensitivity of $n \Delta$ for median.
Accordingly, standard mechanisms would have noise magnitude of $\frac{n \Delta}{\epsilon}$, which is significantly worse than both our framework and smooth sensitivity.

It is important to acknowledge that our Sensitivity-Preprocessing framework will not outperform smooth sensitivity in general.  For example, consider again the domain where all data points are bounded in $[0,1]$, and consider the database $D$ of all 1's.  Then our $g(D)=1$, and the Laplace Mechanism would output $med(D)$ with noise of magnitude $\frac{1}{\epsilon n}$.  However, the smooth sensitivity of this database will be $e^{-\epsilon n /2}$, and the smooth sensitivity framework only requires noise of magnitude $e^{-\epsilon n /2} / \epsilon$.  More generally, smooth sensitivity will often do better if most data entries are exponentially close to one value \footnote{It is also necessary to mention neither mechanism will necessarily outperform the other when the assumptions are not fulfilled, and for databases with high local sensitivity both mechanisms will have poor accuracy but in different ways. Consider again the domain where all data points are bounded in $[0,1]$, and let the database $D$ consist of $\frac{n+1}{2}$ values 1 and $\frac{n - 1}{2}$ values 0, so $med(D) = 1$. Our \sensitivity~will output $g(D) \approx 1/2$ and add noise with magnitude $\frac{1}{\epsilon n}$, while smooth sensitivity will add prohibitively large noise parameter of scale $1/\epsilon$ to $med(D)$, and neither gives any accurate information on the true median value.}. To achieve benefits from both techniques, an analyst could simply apply the smooth sensitivity framework after our preprocessing step.  Our preprocessing algorithm will only improve the smooth sensitivity parameters, so this approach will continue to achieve the strong accuracy guarantees of smooth sensitivity on highly concentrated databases \footnote{We note that this can be done efficiently, as both smooth sensitivity and our preprocessing step take time $O(n^2)$ on database-ordered functions.
Database-ordered functions will be defined in Section~\ref{sec:efficient_examples}, and it will be seen that this general class can be implemented in $O(n^2)$ time for our framework. While it is outside the scope of this paper, it is straightforward to see that this also holds for the smooth sensitivity framework and that this property is preserved when applying our preprocessing to the median function. We leave formal proofs of these to future work.}.
On the example above, if we first apply our \sensitivity~and then add noise based on the smooth sensitivity, then we will also have noise that scales approximately as $e^{-\epsilon n / 2} / \epsilon$.  Since our algorithm is a preprocessing step, it is compatible with all techniques for improving accuracy of differentially private algorithms.  We view this as an exciting avenue for future work, to optimize the use of each tool under different parameter settings.

\paragraph{Mean.} We first note that mean is not naturally defined on the empty set, so we define it to be an input parameter $\hmu$ chosen by the data analyst as the estimated mean.  The analyst's choice of $\hmu$ should reflect her prior knowledge, and will play a role in our accuracy guarantees.  Intuitively, if two databases have means that are exponentially far apart, we cannot hope to output both means accurately.  As such, our \sensitivity~will be accurate on databases with mean reasonably close to $\hmu$.  Our efficient implementation of mean will take $O(n^2)$ time and provide the following guarantees.

%\todo{make restatable}
\begin{restatable}{theorem}{mean}
	\label{thm:mean}
	Let $\mu: \R^{<\mathbb{N}} \rightarrow \R$ be the mean function for the data universe of all finite-length real-valued vectors. For chosen parameters $\hmu$ and $\Delta$, along with any database $D = (x_1,....,x_n) \in \R^{<\mathbb{N}}$, we give $O(n^2)$ time access to a function $g: \R^{<\mathbb{N}} \rightarrow \R$ with sensitivity $\Delta$ such that,
	\[
	|g(D) - \mu(D)| \leq \max\left\{\abs{\mu(D) - \hmu} - \frac{n}{3}\Delta, 0 \right\} + \sum_{i=1}^n \max \left\{\frac{27\abs{x_i - \mu(D)}}{n} - \Delta, 0\right\}.
	\]
	
	Additionally, if we are guaranteed that each  $x_i \in [\hmu + \alpha \Delta, \hmu + (\alpha + n)\Delta]$ for $\alpha \in [-n,0]$, then $g(D) = \mu(D)$
	
\end{restatable}

As was previously mentioned, we claim that our framework gives state-of-the-art accuracy for privately outputting mean when no underlying assumptions are made on the database.
It turns out that the lack of assumptions on the database makes outputting mean privately incredibly difficult, where even Propose-Test-Release was unable to privately output mean with iid assumptions.
Often further assumptions such as data drawn from a normal distribution or other distributions that concentrate well are necessary to guarantee highly-accurate private output of mean.
To our knowledge, the best algorithm to output mean privately when no underlying assumptions are made is the naive algorithm, that simply considers the range $[\hmu - \frac{n}{2}\Delta, \hmu + \frac{n}{2}\Delta]$ of length $n\Delta$, and rounds up or down any value outside of this range.
It is important to note that this range must be chosen independently of the database, as catering the range to the considered database can easily be shown to violate privacy.
Restricting values to a range of $n\Delta$ will then ensure that global sensitivity is at most $\Delta$ and standard mechanisms can be applied from here.
For all databases with values inside the range $[\hmu - \frac{n}{2}\Delta, \hmu + \frac{n}{2}\Delta]$ this will then give accurate output.

Note that our second accuracy guarantee similarly considers databases under which our preprocessing correctly outputs the mean.
It can be immediately seen that the allowable range for values in the database extends beyond the range for the naive algorithm, and can in some ways be seen to double this range.
Essentially, the minimum and maximum values must still be within $n\Delta$ for us to guarantee correctly outputting the mean, but the range under with minimum and maximum values can fall is now doubled.
Given the significance of mean as a statistical metric, we still believe this improvement is of significance and is the first to improve upon the naive algorithm when no underlying assumptions are made with regard to the database.
%Furthermore, this improvement came from a simple black-box application of our framework, and we believe that 

%Note that the constants in the bound should probably be smaller, but were necessary for the analysis.
To complement this comparison to the naive algorithm, we also give strong accuracy guarantees for all databases not just the ones output correctly in our preprocessing.
Unpacking the bound in Theorem \ref{thm:mean}, the first term says that the mean of the database cannot be too far from $\hmu$. The second term considers the individual sensitivity of each data point, where $\abs{x_i - \mu(D)}/n$ is roughly the amount the mean changes from adding $x_i$ to the database. The sensitivity bound on $g$ requires that each individual change can only be offset by an additive $\Delta$, and we need to consider this contribution from each input. Intuitively, our error is small for databases whose mean is reasonably close to $\hmu$ and do not have many significant outliers, which is exactly what one would expect.

\paragraph{Variance.} Variance is also not naturally defined on the empty set, so we define it to be 0 for simplicity. Our efficient implementation of variance takes $O(n^2)$ time and provides the following guarantee.
\begin{restatable}{theorem}{variance}
	\label{thm:variance}
	Let $\textbf{Var}: \R^{<\mathbb{N}} \rightarrow \R$ be the variance function for the data universe of all finite-length real-valued vectors. For fixed parameter $\Delta$, along with any database $D = (x_1,....,x_n) \in \R^{<\mathbb{N}}$, we have $O(n^2)$ time access to a function $g: \R^{<\mathbb{N}} \rightarrow \R$ with sensitivity $\Delta$ such that,
	\[
	|g(D) - \var{D}| \leq \max\left\{\var{D} - \frac{n}{2}\Delta,0\right\} + \sum_{i=1}^n \max \left\{\sum_{j=1}^n \frac{4(x_i - x_j)^2}{n^2} - \Delta, 0\right\}.
	\]
\end{restatable}
%\todo{make restatable}

The primary takeaway from the bound in Theorem \ref{thm:variance} is that databases with reasonably small variance and no major outliers will have low error bounds.  The first term in the error bound says that the variance of our database cannot be too large, which follows from our choice of the empty set to be defined at 0. 
The second term is a bit more messy, but has a natural interpretation under the known deformulation of variance, where we can consider $\sum_{j=1}^n(x_i - x_j)^2/n^2$ to be the contribution of input $x_i$ to the variance. This contribution can then be offset by $\Delta$, and we need to consider this contribution from each input.

\subsubsection*{Personalized Privacy}

%\todo{write this}

Due to the preprocessing aspect of our \sensitivity, we can also apply our framework to variant definitions of differential privacy.
In particular, we consider \emph{personalized differential privacy} introduced in \cite{GR13}, which allows for a more refined definition whereby each individual receives their own $\epsilon_i$ privacy parameter.
Our definition of \emph{individual sensitivity} is then motivated by this privacy variant, as it can be exactly seen as the complementary sensitivity measure for this variant.
More specifically, most privacy mechanisms add noise proportional to $\Delta f / \epsilon$ for outputting functions $f: \mD \to \mathcal{R}$ while still preserving $\epsilon$-differential privacy.
This intimate connection between $\Delta f$ and $\epsilon$ in the output accuracy will be equivalent for the individual sensitivity measures $\Delta_i(f)$ and its respective $\epsilon_i$.
Consequently, the necessary noise for personalized privacy will be proportional to $\max_i \Delta_i(f) / \epsilon_i$.
We will formally prove this fact for two of the most fundamental mechanisms, Laplace and Exponential, and further remark (Remark \ref{rem.expo}) that this approach extends to any $\eps$-differentially private mechanism.

%which suggests that it can be extended to most private mechanisms.

%\todo{informal thm here}

\begin{theorem}[Informal version of Propositions \ref{prop.laplace} and \ref{prop.expo}]
	For both the Laplace and Exponential Mechanisms, instead of adding noise proportional to $\Delta / \epsilon$, the added noise can be proportional to $\max_i \Delta_i /\epsilon_i$ to ensure personalized differential privacy for privacy parameters $\{\epsilon_i\}$.

\end{theorem}

As a result, it is no longer necessarily optimal to set the individual sensitivity parameters $\{\Delta_i\}$ in our \sensitivity~to be equal, but instead set them according to the given $\{\epsilon_i\}$ privacy parameters towards the goal of having each $\Delta_i(g) / \epsilon_i$ roughly equal.
This extends the interest in our \sensitivity~beyond the context of dealing with worst-case sensitivity being much greater than average sensitivity.

For example, consider a well-behaved function where $\{\Delta_i\} = \Delta$, and $\{\epsilon_i\} = \epsilon$ for all $i$ except for some individual $j$ where $\epsilon_{j} = \epsilon/2$. Under this situation it may instead be optimal to halve the individual sensitivity of $j$, which will allow adding half as much noise while only incurring a small additive error by restricting the sensitivity of just one person.

In addition, the error bounds from our general procedure allow for the intuitive fact that increasing any individual sensitivity will increase the accuracy of our preprocessing step for databases including that individual.
We note that when trying to preserve the fraction $\Delta_i(g) / \epsilon_i$, any increase in $\epsilon_i$ (reduced privacy for individual $i$) will allow us to increase our $\Delta_i$ parameter, improving accuracy as desired.
In this way, our \sensitivity~is able to fully take advantage of the heterogeneous $\epsilon_i$ in the variant definition of personalized privacy, and is the first to give accuracy bounds that increase/decrease independently with respect to each $\epsilon_i$.

We believe that these accuracy guarantees can be of further interest in the context of markets for privacy, where individuals sell their data to an analyst and demand different amounts of privacy, represented by their respective $\epsilon_i$.
The trade-off between privacy and accuracy is naturally formalized in these markets through the analyst's budget for procuring accurate estimates of population statistics.
Applying our \sensitivity~to achieve individualized privacy guarantees will allow the analyst to more optimally balance these trade-offs because the accuracy will respond proportionally to changes in privacy for each individual.

\subsubsection*{Higher-Dimensional Extensions for $\ell_1$ Sensitivity}

Our \sensitivity~was only defined for 1-dimensional $f: \mD \to \R$.  We also consider the setting where $f: \mD \to \R^d$.
We note that our \sensitivity~could instead be given parameters $\{\Delta_i\}$ where $\Delta_i = (\Delta_{i,1},...,\Delta_{i,d})$ has different sensitivity parameters for each dimension of the function.
We could then apply our \sensitivity~to each dimension independently and would achieve the corresponding bounds on sensitivity.
However, this approach would require adding noise independently to each dimension when applying a differentially private mechanism on the \sensitivity.
Instead, we would like to only require bounds on the $\ell_1$ sensitivity of our constructed function.

We give a natural extension of our \sensitivity~to higher dimensions, and show that the accuracy guarantees continue to hold in $\ell_1$-distance when $f$ is 2-dimensional (Theorem \ref{thm:main_higher_dim}).  We also show that this construction fails to extend to higher dimensions because a key fact about the intersection of $\ell_1$ balls only holds in 1 and 2 dimensions.

%There is a natural way to define an extension of our \sensitivity~to higher dimensions, but unfortunately this extension will only work for 2-dimensions. 
%Interestingly, the ability to extend our \sensitivity~to higher dimensions will boil down to a fact about the intersection of $\ell_1$ balls that only holds in 1 and 2 dimensions. 

% !TEX root = intro.tex

\subsection{Related Work}

Our work touches upon several areas of interest.  We first discuss previous work on dealing with outlying databases within the data universe, and then discuss previous work on personalized differential privacy and its use within the markets for privacy literature.

\subsubsection*{Worst-case vs average-case sensitivity}

Instance-specific noise for dealing with worst-case sensitivity was first introduced in \cite{NRS07}, where they considered adding noise proportional to \emph{local sensitivity}.
In order to avoid leaking too much information through noise added by local sensitivity, \cite{NRS07} constructed a \textit{smooth sensitivity} metric that minimized the instance-specific noise while still ensuring differential privacy.
They further showed that smooth sensitivity could be efficiently computed and utilized for a variety of important functions for which average sensitivity was much smaller than \textit{global sensitivity}. However, for some functions computing smooth sensitivity was NP-hard or even uncomputable, which inspired the introduction of Sample-and-Aggregate, a technique that preserved privacy and was efficient on all functions with bounded range and for sufficiently large databases.
The general idea was to approximate the function with random subsets of the given database in order to impose stronger bounds on the sensitivity of this approximation.
Combining this with smooth sensitivity allowed for strong error guarantees under the assumption that random subsets of the database often well-approximated the full database.

In order to avoid some of these assumptions, an alternate framework, Propose-Test-Release, was provided in \cite{DL09}, which also heavily relied on the notion of local sensitivity.
In particular, their framework would check if a given database was ``far away" from an outlier, and only release an accurate estimate of the output under this specific circumstance, while outputting $\bot$ (meaning {\sc Null}) otherwise. 
Furthermore, this algorithm would define the outlying databases by explicitly setting the allowable upper bound on local sensitivity.
They show how to implement this framework efficiently for several important functions, and give strong error guarantees when the mechanism does not output $\bot$.

Both of these frameworks relied upon local sensitivity, which is still a worst-case metric.  It is possible for most databases to have high local sensitivity while still having small average sensitivity.
To remedy this issue, previous work instead considered fitting the original function to one with global sensitivity closer to the average sensitivity. 
This preprocessing step can largely be thought of as forcing the output of outlying databases to be closer to that of the well-behaved databases.
This procedure then fits in the general notion of Lipschitz extensions.
Informally, Lipschitz extensions show that there always exists an extension of a smooth function restricted to a subspace to the entire metric space.
By considering ``smoothness'' in the context of differential privacy to be the sensitivity of the function, previous work generally considers the restricted subspace to be the well-behaved databases. 

%\begin{definition}[Lipschitz extension]
%	Given a function $f: X \to Y$ where $X$ and $Y$ are metric spaces. Let $f|_A$ be the function restricted to some $A \subseteq X$. If
%	
%	$\dd_X(x,a) \leq L \dd_Y(f(x))$
%	
%\end{definition}
% which consider a function $f: X \to Y$ where $X$ and $Y$ are general

Lipschitz extensions were first implicitly used in \cite{KKRS13} under the context of \emph{node} differential privacy. 
This work considered restricting the maximum degree of graphs for outputting a variety of graph statistics in bounded-degree graphs. 
This work was then extended in \cite{RS16} which gave efficient Lipschitz extensions for higher-dimensional functions on graphs such as degree distribution. 
In this work, \cite{RS16} further utilize Lipschitz extensions for a generalization of the exponential mechanism.  Lipschitz extensions were also considered in \cite{BBDO13} were the goal was to achieve a \emph{restricted sensitivity} under a certain hypothesis of the database universe and extending to the entire data universe with this global sensitivity constraint. While this procedure was in general computationally inefficient, \cite{BBDO13} gave efficient versions for subgraph counting queries and local profile queries.

Our technique of considering only strictly smaller neighboring databases is related to a technique used to achieve differential privacy over graphs.  The \emph{down sensitivity} \cite{RS15} (also called \emph{empirical global sensitivity} in \cite{CZ13}) of a function at a graph $G$ is the global sensitivity of the function when restricted to the space of all subgraphs of $G$.  That is, it is the maximum change in the function's value between any two neighboring subgraphs of $G$.  Similar to our work, this requires checking sensitivity on a smaller number of neighboring databases, and can allow less noise to be added to analysis on databases with small down sensitivity. Through this lens, our construction of \sensitivity~can be viewed as ensuring that all databases have low down sensitivity. However, an important distinction between these two results is that down sensitivity considers all pairs of neighboring subgraphs of $G$, which, for example, may be the empty graph and a single node for a large graph $G$.  To contrast, at each recursive step of our algorithm, we only consider only smaller neighbors of the current database, i.e., with one entry removed.  This refined analysis means that a database might have large down sensitivity, and our \sensitivity~can still be accurate

\subsubsection*{Personalized privacy and markets for privacy}

%In this result, we will also apply our framework to a variant definition \emph{personalized differential privacy} that was first introduced in \cite{GR13} and gives each individual their own privacy parameter $\epsilon_i$. This definition has subsequently been used in a variety of work toward constructing markets for privacy \todo{cite}.

%Several previous results give mechanisms specific to personalized differential privacy by randomly keeping each individuals data in the database with probability proportional to their respective $\epsilon_i$ \cite{JYC15,LXJJ17}.
%However, they are unable to provide corresponding error guarantees with such procedures for general functions.

We show how our Sensitivity-Preprocessing framework can be applied to \emph{personalized differential privacy}, where each user in the database has her own privacy parameter $\eps_i$.  This definition was first introduced by \cite{GR13}, in the context of purchasing data from privacy-sensitive individuals.  A subsequent line of work on market design for private data \cite{CCKMV13, NOS12, NST12, LR12, FL12, GR13, GLRS14, CLR+15, CIL15, WFA15, CPWV16} leveraged personalized privacy guarantees to purchase data with different privacy guarantees from individuals with heterogeneous privacy preferences.  The vast majority of this work focused on the market design problem of procuring data, and not on the differentially private algorithms that provided personalized privacy guarantees.  \cite{CLR+15} gave a technique for achieving personalized privacy for linear functions by reweighting each person's data inversely proportional to their privacy guarantee. Unfortunately, this reweighting technique does not extend beyond linear functions.  \cite{CCKMV13} proposed an even stronger notion of personalized privacy, that was both personalized and data-dependent, but did not give any algorithmic techniques to satisfy this definition.

Several other results gave mechanisms specific to personalized differential privacy by randomly keeping each individuals data in the database with probability proportional to their respective $\epsilon_i$ \cite{JYC15, AGK17, LXJJ17}.
However, they are unable to provide corresponding error guarantees with such procedures for general functions. \cite{AKZ+17} gave a technique for providing two-tiered personalized privacy guarantees.  Some users received differential privacy and some users received a stronger guarantees of \emph{local differential privacy}, where the users do not trust the data analyst to see their true data.

Finally, there is a small body work on high probability privacy guarantees and average-case privacy guarantees \cite{BLR08, CLN+16, BF16}.  This work addresses a very different problem than we study here.  These papers assume that databases are sampled according to some distribution over the data universe, and provide high probability guarantees with respect to the sampling distribution, allowing a failure of either privacy or accuracy on some set of unlikely databases.  To contrast, we assume that databases are fixed, not randomly sampled.  We provide privacy and accuracy guarantees that depend on the well-behavedness of a given function over the data universe, and our guarantees hold everywhere in the data universe.

\subsection{Organization}
%\todo{ordering?}
In Section~\ref{s.prelims}, we introduce some of the notation and basic definitions that will be used throughout the paper.
In Section~\ref{sec:mainAlgo}, we introduce our \sensitivity~and prove its general accuracy guarantees.
In Section~\ref{s.compute}, we give optimality and hardness guarantees for our sensitivity-preprocessing procedure.
In Section~\ref{sec:efficient_examples}, we show that several important functions can be efficiently implemented in our framework, such as mean, median, maximum, minimum, and we also give strong error guarantees on the implementation of mean.
In Section~\ref{sec:variance_algo}, we efficiently implement our framework for variance and give corresponding error guarantees.
In Section~\ref{s.personal}, we prove several useful facts regarding individual sensitivity and the variant definition of personalized differential privacy, and show how our framework can be very useful in this context.
In Section~\ref{sec:higher_dimensions}, we consider a natural extension of our algorithm that bounds a function's sensitivity in the $\ell_1$ metric for 2 dimensions.

% !TEX root = main.tex

\section{Preliminaries}\label{s.prelims}

%$f$ is function, $g$ is fitting function (possibly rename), $D$ is database, $g(D + x_i)$ is database with element $x_i$ added. similar for $D-x_i$.

We introduce the standard notion of differential privacy and the corresponding global sensitivity metric.  We say that two databases are \emph{neighboring} if they differ in at most one entry.

\begin{definition}[Differential privacy \cite{DMNS06}]\label{def.dp}
	A mechanism $\cM: \mathcal{D} \rightarrow \mathcal{R}$ is \emph{$\epsilon$-differentially private} if for every pair of neighboring databases $D, D' \in \mathcal{D}$, and for every subset of possible outputs $\mathcal{S} \subseteq \mathcal{R}$,
	\[ \Pr[\cM(D) \in \mathcal{S}] \leq \exp(\epsilon)\Pr[\cM(D') \in \mathcal{S}]. \]
\end{definition}

\begin{definition}[Global Sensitivity]\label{def.globalsens}
	The \emph{global sensitivity} of a function $f: \cD \to \mathbb{R}^d$ is:
	\[ \Delta f = \max_{D, D', \; neighbors} \norm{ f(D) - f(D') }_1. \]
\end{definition}

Our result is primarily concerned with tailoring a more refined version of global sensitivity, for which we will need more specific notation for neighboring databases. 
In particular, we will consider the data universe $\mD$ to be composed of (a possibly infinite) collection of individuals, where $x_i$ will denote the data of individual $i$. Any database $D \in \mD$ is then composed of the data of some finite subset of individuals $I$, so $D = \{x_i: i \in I\}$. For ease of notation, we will often assume that $D = (x_1,...,x_n)$. Further, if individual $i$'s data is contained in database $D$, then we will consider $D - x_i$ to be the database with individual $i$'s data removed. Similarly, if $i$'s data is not included in database $D$, then we consider the database $D + x_i$ to be the database with $i$'s data added. We will also sometimes use $i \in D$ to denote that individual $i$'s data is included in database $D$.
Finally, we also assume that for any $D \in \mD$, if $D' \subset D$ is non-empty, then $D' \in \mD$.

With this notation, we introduce the notion of individual sensitivity that is the maximum change in output that is possible by adding individual $i$'s data.

%\rc{need to be more precise about this definition. clearer with connection to i-th element being i-th person.}
\begin{definition}[Individual Sensitivity]
The \emph{individual sensitivity} of a function $f:\mathcal{D} \rightarrow \mathbb{R}^d$ with respect to $i$ is:
%		Given any function $f:\mathcal{D} \rightarrow \mathbb{R}^d$, let the \emph{individual sensitivity} of $f$ with respect to $i$ be	
		\[
		\Delta_i(f) \defeq \max_{x_i,\{D: i\notin D\}}\norm{f(D) - f(D + x_i)}_1.
		\]
We further let $\{\Delta_i(f)\}$ denote the individual sensitivities of $f$ to all individuals.% in the data universe $\mD$.
\end{definition}

For reference, we also provide the definition of local sensitivity that will not be used in this work, but was referred to extensively in related works.

\begin{definition}[Local Sensitivity]\label{def.localsens}
	The \emph{local sensitivity} of a function $f: \cD \to \mathbb{R}^d$ at database $D \in \cD$ is:
	\[ \Delta_D f = \max_{D': \; \text{neighbor of $D$}} \norm{ f(D) - f(D') }_1. \]
\end{definition}

% !TEX root = main.tex

\section{Sensitivity-Preprocessing Function}\label{sec:mainAlgo}

%\todo{Organization of this section}

In this section we formally define our Sensitivity-Preprocessing Function, give the corresponding constructive algorithm for accessing this function, and prove instance-specific error bounds between the original function and our \sensitivity.
Recall that our primary goal is to give an alternate schema for fitting a general function to a sensitivity bounded function.
More specifically, suppose we are given a function $f: \mD \to \bR$ and desired sensitivity parameters $\{\Delta_i\}$, and want to produce another function $g: \mD \to \bR$ that closely approximates $f$, and has individual sensitivity at most $\Delta_i$ for all $i$.

The \sensitivity~will ultimately be defined as a simple greedy recursion that builds up from the empty set.
The key insight is that we can take advantage of the particular metric space structure of databases such that defining our function on a new database only depends on the subsets of that database.
We first use the fact that while each database could have infinitely many neighboring databases, it only has a finite amount of neighbors with strictly fewer entries.
This will allow us to only consider the constraints incurred by each $D - x_i$ for some database $D$.
In particular, for each $g(D - x_i)$ it is allowable to place $g(D)$ anywhere in the region $[g(D - x_i) - \Delta_i, g(D - x_i) + \Delta_i]$.
Intersecting each of these intervals will give the feasible region for $g(D)$, and we will greedily chose the point closest to $f(D)$.
We then use the fact that any two neighbors of a strictly larger database must also be neighbors of a strictly smaller database.
This will ensure that the intersection of all feasible intervals is non-empty, even under our greedy construction.

As a result, the \sensitivity~$g$ is defined inductively starting from the empty set, and new data points are added one by one.  The algorithm ensures that the value of $g$ changes by at most $\Delta_i$ when new data point $x_i$ is added, while minimizing the distance $|f(D) - g(D)|$ at every point.

%\rc{old paragraph:} Instead of simply choosing the new sensitivity, we give a generalization that chooses the sensitivity of each individual. This generalization will be critical for our application to markets for privacy. Further, this will not increase the running time, and it is easy to see that this will also choose the overall sensitivity as the maximum individual sensitivity is the overall sensitivity 

%\todo{dont have $\{\Delta_i\}$ before the definition?}
\begin{definition}[Sensitivity-Preprocessing Function]\label{def:1D_sensitivity_function}
	Given any function $f: \mathcal{D} \rightarrow \R$ and non-negative parameters $\{\Delta_i\}$, we say that a function $g: \mathcal{D} \rightarrow \R$ is a \sensitivity~of $f$ with parameters $\{\Delta_i\}$, if $g(\emptyset) = f(\emptyset)$\footnote{See Remark \ref{rem.empty} for a discussion of how to initialize $g(\emptyset)$ if $f(\emptyset)$ is not well-defined.} and 
	
	\[g(D)= 
	\begin{cases}
	\upper(D),& \text{if } \upper(D)\leq f(D)\\
	\lowerbound(D),& \text{if } \lowerbound(D)\geq f(D)\\
	f(D),              & \text{otherwise}
	\end{cases} \]
	where $\upper(D) = \min_{j \in D} \{g(D- x_j) + \newSens_j\}$ and $\lowerbound(D) = \max_{j \in D} \{g(D- x_j) - \newSens_j\}$.
	
	If $\{\Delta_i\} = \Delta$ for some non-negative $\Delta$, then we say that $g$ is a \sensitivity~of $f$ with parameter $\Delta$.
	
\end{definition}

The generalization from global sensitivity to individual sensitivities is critical to our personalized privacy results in Section \ref{s.personal}.  This generalization does not increase the running time, and it is easy to see that this yields global sensitivity equal to the maximum individual sensitivity.

\subsection{Algorithmic Construction of Sensitivity-Preprocessing Function}\label{subsec:one-dim_algo}

The algorithm \fitting~is presented in Algorithm \ref{algo.fitting}.  It begins by initializing $g(\emptyset)=f(\emptyset)$, and building up $g$ to be defined on databases of increasing size.  At each step, the algorithm ensures that no sensitivity constraints are violated and chooses the best value for $g(D)$ subject to those constraints.

For a given database $D$, $\upper(D)$ is the maximum value $g(D)$ can take without letting it increase too much from a smaller database (violating an individual sensitivity parameter). Similarly, $\lowerbound(D)$ is the minimum value we can make $g(D)$ without letting it decrease too much from a smaller database.  We then define $g(D)$ to be the value in $[\lowerbound(D),\upper(D)]$ that is the closest to $f(D)$.

%\todo{different name for algorithm}

%\todo{Make it clear that this is the same $g$ for all databases we may wish to query.  This is clear from thm statement but not in algo.}

%Our algorithm should be thought of as querying a fixed function $g$ with individual sensitivities $\{\Delta_i\}$, evaluated on the input database $D^*$.  Running the algorithm on different databases with the same sensitivity parameters $\{\Delta_i\}$ will query the same function at a different point. \rc{this is shitty. edit it.}

%The algorithm will not actually do anything all that clever, and just build up from $f(\emptyset)$, ensuring that no sensitivity constraints are violated and choosing the best output among those constraints. 

%\rc{Do we actually only output $g$ evaluated on database of interest?  It feels like we should be outputting more.}

\begin{algorithm}[h!]
	\caption{Sensitivity-Preprocessing Function Algorithm : \fitting($f:\mD \to \bR, \{\Delta_i\}, D$)}
	\begin{algorithmic}
		\State \textbf{Input:} Function $f:\mD \to \bR$, individual sensitivity bounds $\{\Delta_i\}$, and database $D$ of size $n$.
		\State \textbf{Output:} $g(D)$, where $g$ satisfies individual sensitivity $\Delta_i$ for all $i$.
		\State Initialize $g(\emptyset) = f(\emptyset)$
		\For {k=1, \ldots, n}
		\For {every database $D' \subseteq D$ of size $k$}
		\State Set $\upper(D') = \min_{i \in D'} \{g(D'- x_i) + \newSens_i\}$
		\State Set $\lowerbound(D') = \max_{i \in D'} \{g(D'- x_i) - \newSens_i\}$
		\State Set $g(D')= 
\begin{cases}
\upper(D'),& \text{if } \upper(D')\leq f(D')\\
\lowerbound(D'),& \text{if } \lowerbound(D')\geq f(D')\\
f(D'),              & \text{otherwise}
\end{cases}$
		\EndFor
		\EndFor
		\State Output $g(D)$		
	\end{algorithmic}\label{algo.fitting}
\end{algorithm}

%The input will be any function $f:\mathcal{D} \rightarrow \mathbb{R}$, along with individual sensitivity parameters $\newSens_1,...,\newSens_n$, and we build our new function $g$ recursively.
%
%Initialize $g(\emptyset) = f(\emptyset)$. Then for any database $D$, we define
%
%\[
%g(D)= 
%\begin{cases}
%\upper(D),& \text{if } \upper(D)\leq f(D)\\
%\lowerbound(D),& \text{if } \lowerbound(D)\geq f(D)\\
%f(D),              & \text{otherwise}
%\end{cases}
%\] 
%
%where 
%
%\[
%\upper(D) = \min_{i \in D} \{g(D- x_i) + \newSens_i\}
%\]
%
%and 
%
%\[
%\lowerbound(D) = \max_{i \in D} \{g(D- x_i) - \newSens_i\}
%\]

%Intuitively, $\upper(D)$ is just the maximum value we can make $g(D)$ without letting it increase too much from a previous database (violating an individual sensitivity parameter). Similarly, $\lowerbound(D)$ is just the minimum value we can make $g(D)$ without letting it decrease too much from a previous database.
%
%The construction of $g$ was entirely to ensure that each individual sensitivity did not exceed $\newSens_i$ for every person $i$. We can then use these bounds for adding noise to $g$ and ensuring differential privacy. 

This construction of $g$ ensures that  the individual sensitivity of $g$ does not exceed $\newSens_i$ for each $i$.  We can then use these bounds on the sensitivity of $g$ to calibrate the scale of noise that must be added to ensure differential privacy.  In the special case that $\Delta_i = \Delta$ for all $i$, then the global sensitivity of $g$ is $\Delta$, and we can add noise that scales with $O(\frac{\Delta}{\eps})$ to achieve $\eps$-differential privacy.  Note that this guarantee holds even if $f$ has unbounded sensitivity.  In Section \ref{s.personal}, we show how to satisfy differential privacy under heterogeneous $\Delta_i$.

\begin{remark}\label{rem.empty}
Our algorithm is initialized using $f(\emptyset)$, and thus centers $g$ around this point.  In the case that $f(\emptyset)$ is undefined---for example, when $f$ computes the mean of a database---the analyst should initialize $g(\emptyset)$ using some domain knowledge or prior beliefs on reasonable centering of the function.  If no prior knowledge is available, the analyst can sample multiple databases and evaluate $f$ on the samples to estimate a reasonable centering point for $g(\emptyset)$.  The sensitivity bounds will still hold regardless of the centering of $g$, but accuracy may suffer if $g(\emptyset)$ is set to be far from most values of $f$. 
%\rc{When $f(\emptyset)$ undefined, either know something about the database/function and define cleverly, or if you don't know anything, then make shit up. sensitivity always works, but accuracy may be terrible if you guess wrong.  have an aside about this. could also sample a bunch of databases and see where their output lies.}
\end{remark}

In addition to sensitivity guarantees and runtime analysis, we also provide an instance-specific error bound. 
Unfortunately this bound will not be in a clean form, but it does capture the intuitive fact that if we increase any $\Delta_i$ then it is likely that accuracy also increases.

However, we are able to obtain a bit more intuition on our instance-specific error bounds, and can consider them in a similar context to local sensitivity.
Given that our \sensitivity~defines a database recursively in terms of its subsets, it makes sense that our error guarantees will be in terms of these subsets.
These error bounds can then be seen as capturing the sensitivity between the neighboring subsets of $D$. 
Analogously to local sensitivity, we will have larger errors for databases with high sensitivity between the neighboring subsets.

%\todo{Interpretation of theorem} But we also wanted to keep $g$ as close to $f$ as possible. Due to our 'centering' $g$ around $f(\emptyset)$, we will maintain high accuracy around $f(\emptyset)$, but this accuracy can start to degrade if we need to make substantial changes (exceeding the respective $\newSens_i$) to get to a larger database. %The stronger accuracy guarantees will be stated in these terms, and the corollary afterwards will give the weaker, but more easily digestible guarantees. \rc{what corollary? this may be outdated text}

%\rc{$D^*$ notation from algo may be annoying in Thm \ref{thm:main_1d}.  We can change.}

\begin{theorem}\label{thm:main_1d}
	Given $T(n)$ time query access to an arbitrary function $f:\mathcal{D} \rightarrow \mathbb{R}$, and sensitivity parameters $\{\Delta_i\}$, \fitting~provides $O((T(n) + n)2^n)$ time access to the \sensitivity~$g:\mathcal{D} \rightarrow \mathbb{R}$ such that $\Delta_i(g) \leq \newSens_i$ for all $i$.  Further, for any database $D = (x_1,...,x_n)$,
	\[ 
	\abs{f(D) - g(D)} \leq \max_{\sigma \in \sigma_D}\sum_{i=1}^{|D|} \max \{ \abs{f(D_{\sigma(<i)} + x_{\sigma(i)}) - f(D_{\sigma(<i)})} - \newSens_{\sigma(i)}, 0 \},
	\]
where $\sigma_D$ is the set of all permutations on $[n]$, and $D_{\sigma(<i)} = (x_{\sigma(1)},...,x_{\sigma(i-1)} )$ is the subset of $D$ that includes all individual data in the permutation before the $i$th entry.
\end{theorem}

%\todo{additional result saying that this algo queries the same $g$ even for different input $D^*$?}
%Theorem \ref{thm:main_1d} is proven in the remainder of this section.

\begin{remark}\label{rem:higher_dim}
	We can easily extend this theorem to $f: \mD \to \R^d$ by running \fitting~on each dimension independently in terms of sensitivity parameters and error bounds
	Specifically, suppose we were instead given parameters $\{\Delta_i\}$ where $\Delta_i = (\Delta_{i,1},...,\Delta_{i,d})$ has different sensitivity parameters for each dimension of the function.
	We could then consider the function restricted to a single dimension $d'$, and run \fitting~on this projection with sensitivity parameters $\{\Delta_{i,d'}\}$.
	This will give the desired sensitivity bounds in that single dimension, then running \fitting~on all dimensions and composing across dimensions will give the appropriate \sensitivity~in $d$ dimensions.  In Section \ref{sec:higher_dimensions},  we consider extensions to higher dimensions where each dimension is not treated independently.
	
\end{remark}

%\todo{corollary about higher dimensions and treating them all independently}

%\david{Not necessarily the best ordering, but I just wanted to put it in the paper for now.}

\subsection{Sensitivity-Preprocessing Function Correctness}

We first prove that the \sensitivity~given in Definition~\ref{def:1D_sensitivity_function} both meets the individual sensitivity criteria and is also defined on all databases.

\begin{lemma}\label{lem:1d_sensitivity_guarantees}
For any function $f:\mathcal{D} \rightarrow \mathbb{R}$ and non-negative sensitivity parameters $\{\Delta_i\}$, if $g:\mathcal{D} \rightarrow \mathbb{R}$ is defined according to Definition~\ref{def:1D_sensitivity_function}, then $g$ is defined on all databases $D \in \mathcal{D}$ and $\Delta_i(g) \leq \newSens_i$ for all $i$.
\end{lemma}

\begin{proof}
	It suffices to show that for any $D \in \mathcal{D}$ with at least one entry and for any $x_i \in D$, we have $g(D - x_i) - \newSens_i \leq g(D) \leq g(D - x_i) + \newSens_i$. By our construction, this must always be true if $\lowerbound(D) \leq g(D) \leq \upper(D)$ for any $D \in \mathcal{D} \setminus \{\emptyset \}$. Our construction of $g$ will always place $g(D) \in [\lowerbound(D),\upper(D)]$ if the interval is non-empty, so it suffices to show that for all $D \in \mathcal{D} \setminus \{ \emptyset \}$, 
	\[
	 \lowerbound(D) \leq \upper(D).
	\]

	We will prove this by induction starting with $D = x_i$ with one entry. Therefore, $\upper(D) = f(\emptyset) + \newSens_i$ and $\lowerbound(D) = f(\emptyset) - \newSens_i$, which implies our desired inequality because $\Delta_i \geq 0$.
	
	We now consider an arbitrary $D$ and assume that our claim holds for all $D' \subset D$.
	Let $x_k \in D$ minimize $g(D - x_i) + \newSens_i$ over all $x_i \in D$, so 
	\[
	\upper(D) = g(D - x_k) + \Delta_k,
	\]
	and let $x_j \in D$ maximize $g(D - x_i) - \newSens_i$ over all $x_i \in D$, so 
	\[
	\lowerbound(D) = g(D - x_j) - \newSens_j.
	\] 
	
	If $k = j$ then the desired inequality immediately follows. Otherwise we consider $D - x_k - x_j$. By our inductive hypothesis, we know $\lowerbound(D - x_k) \leq g(D - x_k) \leq \upper(D - x_k)$, so 
	\[
	g(D - x_k) \geq \lowerbound(D - x_k) \geq g(D - x_k - x_j) - \newSens_j.
	\] 
	Similarly, we have $\lowerbound(D - x_j) \leq g(D - x_j) \leq \upper(D - x_j) $, so 
	\[
	g(D - x_j) \leq \upper(D - x_j) \leq g(D - x_k - x_j) + \newSens_k.
	\]
	
	Combining these inequalities gives $g(D - x_k) + \newSens_j \geq g(D - x_j) - \newSens_k$, which implies our desired result.
	
\end{proof}

\subsection{Error Bounds for Sensitivity-Preprocessing Function}

We now prove the desired instance-specific error bounds between the original function and our \sensitivity.

\begin{lemma}\label{lem:1D_error_bounds}
	For any function $f:\mathcal{D} \rightarrow \mathbb{R}$ and non-negative sensitivity parameters $\{\Delta_i\}$, if $g:\mathcal{D} \rightarrow \mathbb{R}$ is defined according to Definition~\ref{def:1D_sensitivity_function}, then for any database $D \in \mathcal{D}$,
	\[ 
	\abs{f(D) - g(D)} \leq \max_{\sigma \in \sigma_D}\sum_{i=1}^{|D|} \max \{ \abs{f(D_{\sigma(<i)} + x_{\sigma(i)}) - f(D_{\sigma(<i)})} - \newSens_{\sigma(i)}, 0 \},
	\]
	where $\sigma_D$ is the set of all permutations on $[n]$, and $D_{\sigma(<i)} = (x_{\sigma(1)},\ldots,x_{\sigma(i-1)} )$ is the subset of $D$ that includes all individual data in the permutation before the $i$th entry.
	
\end{lemma}

\begin{proof}
%	The sensitivity guarantees are given in the previous section, and the runtime follows immediately from construction. It then remains to show for all $D$ that
%	
%	\[ 
%	\abs{f(D) - g(D)} \leq \max_{\sigma \in \sigma_D}\sum_{i=1}^{|D|} \max \{ \abs{f(D_{\sigma(<i)} + x_{\sigma(i)}) - f(D_{\sigma(<i)})} - \newSens_{\sigma(i)}, 0 \}
%	\]
	
	We will prove this claim inductively and first consider $D = x_j$ with one entry for some $j$.  We need to show 
	\[
	\abs{f(D) - g(D)} \leq \max\{\abs{f(D) - f(\emptyset)} - \newSens_j,   0\},
	\]
	which follows easily from construction of $g$. We now consider an arbitrary $D$ and assume that the claim is true for all $D' \subset D$. From our construction we claim that
	\[
	\abs{f(D) - g(D)} \leq \max_{x_i \in D} \{\abs{f(D) - g(D - x_i)} - \newSens_i, 0 \}.
	\]
	This follows from the fact that if $f(D) = g(D)$ then we must have $\abs{f(D) - g(D - x_i)} \leq \newSens_i$ for all $i$, and otherwise there must be some $x_i \in D$ such that the constraint on $g(D)$ with respect to $\Delta_i$ is tight. Using this fact we can bound $\abs{f(D) - g(D)}$ in the following way:
	\begin{align*}
	\abs{f(D) - g(D)} & \leq \max_{x_i \in D} \{\abs{f(D) - g(D - x_i)} - \newSens_i, 0 \} \\ 
	& = \max_{x_i \in D} \{\abs{f(D) - f(D - x_i) + f(D - x_i) - g(D - x_i)} - \newSens_i, 0 \} 
	\\
	& \leq \max_{x_i \in D} \{\abs{f(D) - f(D - x_i)} - \newSens_i + \abs{f(D - x_i) - g(D - x_i)},0  \}
	\\
	& \leq \max_{x_i \in D} \{\max\{\abs{f(D) - f(D - x_i)} - \newSens_i,0\} + \abs{f(D - x_i) - g(D - x_i)}  \}
	\end{align*}

We then apply the inductive hypothesis to $\abs{f(D - x_i) - g(D - x_i)}$, which immediately implies our desired bound.
\end{proof}

\subsection{Proof of Theorem~\ref{thm:main_1d}}

\begin{proof}[Proof of Theorem~\ref{thm:main_1d}]
	The individual sensitivity guarantees are given by Lemma~\ref{lem:1d_sensitivity_guarantees}, and the error bounds are given by Lemma~\ref{lem:1D_error_bounds}. It then remains to show the running time. If we assume $T(n)$ time access to $f$ for a database with $n$ entries, then because we need to query each subset of $D$, this will contribute time $O(T(n)2^n)$. Furthermore, for each subset we need to compute $\upper(D)$ and $\lowerbound(D)$ which takes $O(n)$ time for each subset. This then gives our full runtime of $O((T(n) + n)2^n)$.
	%\todo{running time in the thm statement}

\end{proof}

% !TEX root = main.tex

\section{Optimality and Hardness of Sensitivity-Preprocessing Function}\label{s.compute}

Our algorithm in Section \ref{sec:mainAlgo} took exponential time to query the \sensitivity~$g$ at each database $D$ of interest, and, while we did achieve bounds on the error incurred, their complicated formulation makes it difficult to determine whether these bounds are strong. In this section we give strong justification for our construction of the \sensitivity~in terms of both error incurred and the exponential running time for the general setting.

In Section~\ref{subsec:2-approximation_1d}  we consider the general problem of approximating an arbitrary function $f: \mathcal{D} \rightarrow \R$ with one that has individual sensitivity bounded by $\{ \Delta_i\}$.
Under the $\ell_\infty$ metric, our \sensitivity~will achieve a 2-approximation of the optimal function.
Furthermore, this 2-approximation can still be obtained when the optimal function is restricted to certain subsets of the data universe.
Informally, this will imply that on subsets which allow for small error between $f$ and a function with individual sensitivity bounded by $\{ \Delta_i\}$, our \sensitivity~will also have small error.
Due to $\ell_\infty$ being a worst-case metric, it is then natural to ask if our \sensitivity~actually still performs well on the non-worst-case databases. 
To this end, we show that our \sensitivity~is Pareto optimal, meaning that for any other function with individual sensitivity bounded by $\{ \Delta_i\}$, if it has smaller error on some database relative to our \sensitivity, then there must exist another database on which it has higher error.
%In this section we justify this exponential running time for the general problem by showing that $g$ is hard to compute and hard to approximate. 
 
In Section \ref{s.hardness} we show that it is NP-hard to achieve our approximation guarantees with respect to the $\ell_\infty$ metric.
We further show that it is uncomputable to do better than a 2-approximation in the $\ell_\infty$ metric, and also uncomputable to achieve even a constant approximation in any $\ell_p$ metric for $p < \infty$ which justifies our choice of metric. 
%In Section \ref{subsec:2-approximation_1d} we show that Algorithm \ref{algo.fitting} does indeed find a $g$ that 2-approximates $f$.
We believe that the combination of these results gives a strong indication that our \sensitivity~and corresponding exponential time construction is the best we can hope to achieve for the general problem.
 
% !TEX root = computability.tex

\subsection{Optimality guarantees}\label{subsec:2-approximation_1d}

In this section we prove that our \sensitivity~achieves certain optimality guarantees.
As there are many ways in which to measure how close one function is to another, it is first necessary to be more specific about the definition of optimality we use here.
The set that we are trying to optimize over will be all functions with bounded individual sensitivity:

\begin{definition}
	Given a data universe $\mathcal{D}$ and individual sensitivity parameters $\{ \Delta_i \}$, define
	\[
	F_{\{\Delta_i\}}(\mathcal{D}) \defeq \{f:\mathcal{D} \rightarrow \R \; | \; \Delta_i(f) \leq \Delta_i, \forall i   \}.
	\]
\end{definition}

In this context, the general goal will then be to show that our \sensitivity~is close to the optimal function on this set.
Here we will consider optimal to be under the $\ell_\infty$ metric, where we want $f^* \in F_{\{\Delta_i\}}(\mathcal{D})$ to minimize the maximum difference $\abs{f(D) - f^*(D)}$ over all $D \in \mathcal{D}$.  
Our \sensitivity~achieves a 2-approximation to the optimal $f^* \in F_{\{\Delta_i\}}(\mathcal{D})$ with respect to the $\ell_{\infty}$ metric.
For unbounded sensitivity functions, the value $\abs{f(D) - f^*(D)}$ will be unbounded, so we will instead show the stronger result that this 2-approximation also holds if we restrict the data universe to a single database and its subsets.
Specifically, we show that if for certain subsets of the data universe it is possible to perfectly fit $f$ to a $\{\Delta_i\}$ individual sensitivity bounded function, then our \sensitivity~will also perfectly fit to $f$ in this subset. These guarantees are formalized in the following lemma.

%show that our algorithm gives 2-approximation compared to the $f^*$ with sensitivity at most $\Delta$ that minimizes $\max_{D \in \mathcal{D}} |f(D) - f^*(D)|$ for arbitrary $f$ and data universes $\mathcal{D}$

\begin{lemma}\label{lem:2-approx_1d_key_lemma}
	Given any $f:\mathcal{D} \rightarrow \mathbb{R}$, let $g:\mathcal{D} \rightarrow \mathbb{R}$ be the \sensitivity~of $f$ with parameters $\{\Delta_i\}$. For any arbitrary $D \in \mathcal{D}$, define $\mathcal{D}' = \{D' \subseteq D\}$. Then,
	\[
	\max_{D' \in \mathcal{D}'} |f(D') - g(D')| \leq 2 \min_{f^* \in F_{\{\Delta_i\}}(\mathcal{D}')} \max_{D' \in \mathcal{D}'} |f(D') - f^*(D')|
	\] 	
\end{lemma}

\begin{proof}
	We will prove this inductively on the size of $D$. It is immediately true for $D = \emptyset$. We now prove for arbitrary $D$ where we assume the claim for all strict subsets of $D$. Our proof will be by contradiction, where we suppose that our claim is not true for some $D$.
	
	We first determine the database at which $|f(D') - g(D')|$ is maximized. Suppose $\arg \max_{D' \in \mathcal{D}'} |f(D') - g(D')| = \tilde{D}$ such that $\tilde{D} \subset D$. Define $\mathcal{\tilde{D}} = \{ D' \subseteq \tilde{D} \}$. Because $\tilde{D} \subset D$, it must follow that
	\[
	\min_{f^* \in F_{\{\Delta_i\}}(\mathcal{\tilde{D}})} \max_{D' \in \mathcal{\tilde{D}}} |f(D') - f^*(D')|
	\leq 
	\min_{f^* \in F_{\{\Delta_i\}}(\mathcal{D}')} \max_{D' \in \mathcal{D}'} |f(D') - f^*(D')|.
	\]
	
	By our assumption that the claim is not true on $D$, it follows that
	\[
	|f(\tilde{D}) - g(\tilde{D})| > 2 \min_{f^* \in F_{\{\Delta_i\}}(\mathcal{D}')} \max_{D' \in \mathcal{D}'} |f(D') - f^*(D')|.
	\]
Combining this with the previous inequality implies,
	\[
	|f(\tilde{D}) - g(\tilde{D})| > 2 	\min_{f^* \in F_{\{\Delta_i\}}(\mathcal{\tilde{D}})} \max_{D' \in \mathcal{\tilde{D}}} |f(D') - f^*(D')|,
	\]	
which contradicts our inductive hypothesis. Therefore we must have $\max_{D' \in \mathcal{D}'} |f(D') - g(D')| = |f(D) - g(D)|$.

%	We first claim that if $\arg \max_{D' \in \mathcal{D}'} |f(D') - g(D')| = D$ . This follows from our inductive claim. It is clear that $\max_{D' \subseteq \tilde{D}} |f(D') - f^*(D')| \leq \max_{D' \subseteq {D}} |f(D') - f^*(D')|$ for $\tilde{D} \subseteq D$. Therefore if $|f(\tilde{D}) - g(\tilde{D})| > 2 \max_{D' \subseteq {D}} |f(D') - f^*(D')|$, then we must also have $|f(\tilde{D}) - g(\tilde{D})| > 2 \max_{D' \subseteq \tilde{D}} |f(D') - f^*(D')|$ contradicting our inductive hypothesis if $\tilde{D} \neq D$.
	
	We now apply Lemma~\ref{lem:exists_far_subset_database}, which we prove subsequently, to see that there must exist $\tilde{D} \subset D$ such that $|f(D) - f(\tilde{D})| \geq  |f(D) - g(D)| +  \sum_{i \in D \setminus \tilde{D}} \newSens_i$. Therefore for any $f^* \in F_{\{\Delta_i\}}(\mathcal{D}')$ it must be true that 
	\[
	\max\{|f(\tilde{D}) - f^*(\tilde{D})|,|f(D) - f^*(D)| \} \geq \frac{|f(D) - g(D)|}{2},
	\]
	because of the sensitivity constraints. We then use the fact that $\max_{D' \in \mathcal{D}'} |f(D') - g(D')| = |f(D) - g(D)|$ to conclude,
	\[
	\max_{D' \in \mathcal{D}'} |f(D') - g(D')| \leq 2 \min_{f^* \in F_{\{\Delta_i\}}(\mathcal{D}')} \max_{D' \in \mathcal{D}'} |f(D') - f^*(D')|.
	\] 
This contradicts our assumption, so the claim must therefore be true for $D$.
\end{proof}

\begin{lemma}\label{lem:exists_far_subset_database}
Given any $f:\mathcal{D} \rightarrow \mathbb{R}$, let $g:\mathcal{D} \rightarrow \mathbb{R}$ be the \sensitivity~of $f$ with individual sensitivity parameters $\{ \Delta_i \}$. For any $D \in \mathcal{D}$ such that $f(D) \neq g(D)$ there must exist some $\tilde{D} \subset D$ such that $g(D) \geq f(\tilde{D}) + \sum_{i \in D \setminus \tilde{D}} \newSens_i$ if $f(D) > g(D)$ and $g(D) \leq f(\tilde{D}) - \sum_{i \in D \setminus \tilde{D}} \newSens_i$ if $f(D) < g(D)$.
\end{lemma}

\begin{proof}
	We prove the claim inductively, starting with the immediate observation that by construction it is true when $D$ only has one entry.
	
	We now consider an arbitrary $D$ and assume our claim for all subsets.
	Without loss of generality, we will prove the claim if $f(D) > g(D)$, and can symmetrically apply the proof for the case when $f(D) < g(D)$.
	If $f(D) > g(D)$, then there must exist some $x_i \in D$ such that $g(D) = g(D - x_i) + \newSens_i$. If $f(D - x_i) \leq g(D - x_i)$, then we can set $\tilde{D} = D - x_i$ and the claim follows. Otherwise we must have $f(D - x_i) > g(D - x_i)$ and we apply our inductive hypothesis to obtain some $\tilde{D} \subset D - x_i$ such that 
	\[
	g(D - x_i) \geq f(\tilde{D}) + \sum_{j \in (D - x_i) \setminus \tilde{D}} \newSens_j.
	\] 
	We then use the fact that $g(D) = g(D - x_i) + \newSens_i$ to achieve
	\[
	g(D) \geq f(\tilde{D}) + \sum_{j \in D \setminus \tilde{D}} \newSens_j.
	\] 
\end{proof}

We note that because Lemma~\ref{lem:2-approx_1d_key_lemma} achieves a 2-approximation when the optimal function is restricted to subsets of the data universe, we easily achieve a 2-approximation on the full data universe.

\begin{corollary}\label{cor:2-approx_1d}
	Given any $f:\mathcal{D} \rightarrow \mathbb{R}$, let $g:\mathcal{D} \rightarrow \mathbb{R}$ be the \sensitivity~of $f$ with parameters $\{\Delta_i\}$. Then,
	\[
	\max_{D' \in \mathcal{D}} |f(D') - g(D')| \leq 2 \min_{f^* \in F_{\{\Delta_i\}}(\mathcal{D})} \max_{D' \in \mathcal{D}} |f(D') - f^*(D')|.
	\] 
\end{corollary}

\subsubsection*{Pareto Optimality}

We now complement our localized 2-approximation of the $\ell_\infty$ metric with a Pareto optimality result. 
As $\ell_\infty$ is a worst-case metric we would still like our \sensitivity~to perform well on the non-worst-case databases.
In particular, for the databases that do not contribute to the $\ell_\infty$ error, we still want the error to be minimized.
The following lemma will conclude that we cannot improve the error of a single database without incurring more error on another database, indicating that we are still performing well on the non-worst-case databases.

\begin{lemma}\label{lem:optimality}
	Given any $f:\mathcal{D} \rightarrow \mathbb{R}$, let $g:\mathcal{D} \rightarrow \mathbb{R}$ be the \sensitivity~of $f$ with individual sensitivity parameters $\{ \Delta_i \}$. For any $h \in F_{\{\Delta_i\}}(\mathcal{D})$ if there is some $D \in \mathcal{D}$ such that
	\[ \abs{f(D) - h(D)} < \abs{f(D) - g(D)},\]
	then there also exists some $D' \in \mathcal{D}$ such that 
	\[ \abs{f(D') - h(D')} > \abs{f(D') - g(D')}.\]
\end{lemma}

\begin{proof}
	Suppose there is some $h \in F_{\{\Delta_i\}}(\mathcal{D})$ such that
	\[ \abs{f(D) - h(D)} < \abs{f(D) - g(D)} \] 	
	for some $D \in \mathcal{D}$, and for all $D' \in \mathcal{D}$, 	
	\[ \abs{f(D') - h(D')} \leq \abs{f(D') - g(D')}\]
Then it must be true that $h(\emptyset) = g(\emptyset)$ because $g(\emptyset) = f(\emptyset)$. Let $D$ be the smallest database such that $h(D) \neq g(D)$, which implies that $\abs{f(D) - h(D)} < \abs{f(D) - g(D)}$. 
	This inequality implies $g(D) \neq f(D)$, and by our construction of $g$, either $\upper(D) < f(D)$ or $\lowerbound(D) > f(D)$. 
	
	Without loss of generality, assume $\upper(D) < f(D)$ and thus $g(D) = \upper(D)$. Using the fact that $\abs{f(D) - h(D)} < \abs{f(D) - g(D)}$, we can conclude that $h(D) > g(D)$. 
	However, since $\upper(D) = g(D - x_i) + \Delta_i$ for some $x_i \in D$, we must have $h(D) > g(D - x_i) + \newSens_i$. Our assumption that $D$ was the smallest database such that $h(D) \neq g(D)$ then implies $h(D) > h(D - x_i) + \newSens_i$, contradicting the individual sensitivity of $i$ being at most $\newSens_i$ in $h$. 
	
	Therefore, $F_{\{\Delta_i\}}(\mathcal{D})$ cannot contain such an $h$, which implies our claim.
	
\end{proof}

\subsection{Hardness of approximation}\label{s.hardness}

In this section we justify the exponential running time of our implementation of the \sensitivity~for the general setting.
Recall that in our construction we did not make any assumptions about $\mathcal{D}$ and only required query access to the function $f: \mathcal{D} \rightarrow \R$.
Under this limited knowledge setting it is reasonable that our localized greedy construction is the best we can hope for, despite taking exponential time.
Accordingly, we show here that even if we restrict $\mathcal{D}$ to be exponential-sized, set all $\{\Delta_i\}$ to be the same $\Delta$, and further force $f$ to be polytime representable, it is still NP-hard to compute our \sensitivity.
This proof will further imply that it is NP-hard to compute a function that has identical individual sensitivity guarantees and achieve the same approximation guarantees that our \sensitivity~does in Lemma~\ref{lem:2-approx_1d_key_lemma}.

After proving this NP-hardness result, we will discuss the issues with computing individual sensitivity bounded functions that obtain better approximations.
We give strong justification that it is uncomputable to achieve better than a 2-approximation in the $\ell_\infty$ metric.
Further, we give similar reasons why it is uncomputable to achieve even a constant approximation on average error for the general setting, which justifies our choice of metric for proving our approximation guarantees in the previous section.
We believe these ideas could be formalized in a straightforward manner, but think that doing so is unnecessary for the scope of this paper. 

%In fact, even if the data universe is finite and the function is representable in polynomial many bits, we can still prove the problem is NP-hard.

%\david{I need to write this way better, but just wanted to get it down on the first pass}

\subsubsection*{NP-hardness}

\begin{proposition}\label{prop:NP-hard}
	For certain $f: \mathcal{D} \rightarrow \R$ such that $|\mathcal{D}| = O(3^n)$, it is NP-hard to compute our \sensitivity~$g$  with parameter $\Delta$ on a specific database.
	
\end{proposition}

\begin{proof}
	In order to prove this claim, we will construct a gadget function that takes an arbitrary SAT formula $\phi$ and constructs a function $f: \mathcal{D} \rightarrow \R$ such that $|\mathcal{D}| = O(3^n)$ and on a specified database $D \in \mathcal{D}$, $g(D) < n$ if and only if $\phi$ is satisfiable. We construct that gadget function below.
	
	\paragraph{Gadget Function:} Let $\mathcal{D}$ be the data universe with $n$ individuals such that $x_i \in \{T,F\}$. Let $\phi:\{T,F\}^n \rightarrow \{0,1\}$ be an arbitrary SAT formula of $n$ variables that outputs $0$ if false and $1$ if true. For any $D \in \mathcal{D}$, let $D + T \in \{T,F\}^n$ be the assignment of variables that correspond to $D$ and set all variables not in $D$ to be true. Let the function $f_{\phi}:\mathcal{D} \rightarrow \mathbb{R}$ be defined as $f_{\phi}(D) = |D| - \phi(D + T)$ where $|D| = |\{i\in D\}|$. Further, define $f_{\phi}(\emptyset) = 0$ and let $\Delta = 1$.

	\paragraph{Claim:} For the constructed gadget function $f$ from SAT formula $\phi$ and our corresponding \sensitivity~$g$ with parameter $\Delta$, we must have that $g(F^n) < n$ iff $\phi$ is satisfiable.
	
	\
	
	First, we assume $\phi$ is unsatisfiable which implies $f_{\phi}(D) = |D|$ for all $D$. Therefore the sensitivity of $f$ is $1$, and $g$ will be identical to $f$, so $g(F^n) = n$.
	
	Next, we show that if $g(F^n) \geq n$ then there cannot exist a satisfying assignment of $\phi$. Suppose there does exist a satisfying assignment, then take the one with the fewest false assignments and denote this as $x^* \in \{T,F\}^n$. Further, consider the database $D \subseteq F^n$ that consists of all of the false assignments of $x^*$. By definition, we must have that $f_{\phi}(D) = |D| - 1$, and we further show that $g(D) = |D| - 1$. 
	
	For any $D' \subset D$, we have $f_{\phi}(D') = |D'|$ by construction of $f_{\phi}$ and our assumption that $x^*$ was the satisfying assignment with the fewest false assignments. It is easy to see that $g(D') = |D'|$ by construction, which implies that $g(D) = |D| - 1$. Since the sensitivity is set to be $1$, we have that for every $\tilde{D}$ such that $\tilde{D} \supseteq {D}$ 
	%\rc{what is $\tilde{D}$?} 
	it must be true that $g(\tilde{D}) \leq |\tilde{D}| - 1$. By construction, we know $D \subseteq F^n$, which implies $g(F^n) < n$.  This gives a contradiction and implies that $\phi$ is unsatisfiable.
\end{proof}

Note that to satisfy the approximation guarantees given in Lemma~\ref{lem:2-approx_1d_key_lemma}, any $f^* \in F_{\Delta}(\mathcal{D})$ would require $f^*(D) = |D| - 1$ in our proof as well.
Accordingly, for any $f^* \in F_{\Delta}(\mathcal{D})$ that satisfies the approximation guarantees of Lemma~\ref{lem:2-approx_1d_key_lemma}, it must also be true that $f(F^n) < n$ iff $\phi$ is satisfiable.
Therefore, any algorithm that achieves the same guarantees must also be NP-hard to compute.

%
%\todo{'Show' that getting better than a 2-approx is uncomputable. Also that any finite approximation on the problem}
%
%\david{These might be better to just give high-level intuition and not actually state them as lemmas}

%\begin{lemma}\label{lem:better_than_2approx_uncomp}
%	It is uncomputable to achieve better than a $2$-approximation for computing a $g$ with sensitivity at most $\Delta$ that minimizes $\max_{D \in \mathcal{D}} |f(D) - g(D)|$ for arbitrary $f$ and data universes $\mathcal{D}$
%	
%	
%\end{lemma}

\subsubsection*{Uncomputability of better approximations}

We now argue that it is uncomputable to achieve better approximation factors than our \sensitivity, with respect to both the $\ell_{\infty}$ metric and any $\ell_p$ metric.
%discuss why it would be uncomputable to achieve better approximation factors.

\begin{remark} 
	We claim that no finitely computable algorithm can obtain a function with appropriately bounded individual sensitivities that achieves better than a 2-approximation on the $\ell_\infty$ error.  Let $\mathcal{D}$ only contain the empty set and databases of size one, each containing a single real-valued data entry $x \in [0,1]$, and set $\Delta = 1$. Consider any finite algorithm that constructs a $\Delta$-sensitivity function $h$ to minimize the maximum difference between (adversarially chosen) $f$ and $h$ over all databases. 
	
	If $f$ is arbitrary and only query accessible, then the algorithm can only query a finite number of databases, and an adversary could just set $f(x) = f(\emptyset) = 0$ for all queried databases. In order to achieve even a constant approximation, the algorithm would need to set $f(x) = 0$ just in case $f(x) = 0$ for all $x \in [0,1]$. However, the adversary could then set $f(y) = 2$ for all non-queried databases. The function that minimizes the $\ell_\infty$ error would then set $f(x) = 1/2$ for all queried databases and $f(y) = 3/2$ for all non-queried databases. As a result, the finite algorithm can only achieve a 2-approximation.	
\end{remark}

\begin{remark} 
	We further claim that no finitely computable algorithm can obtain a function with appropriately bounded individual sensitivities that achieves a constant approximation on the average $\ell_p$ error. The optimal function in this scenario would be $f^*$ that minimizes:
	\[
	\min_{f^* \in F_{\{\Delta_i\}}(\mathcal{D})} \left(\frac{\sum_{D \in \mathcal{D}} \left(f(D') - f^*(D')\right)^p}{|\mathcal{D}|}\right)^{1/p}.
	\]
	We consider the same example as above, and note that the number of queried databases is finite and the number of non-queried databases is infinite.
	In order to achieve a constant approximation, the algorithm would need to set $f(x) = 0$ just in case $f(x) = 0$ for all $x \in [0,1]$.
	However, it would then have to set $f(y) = 1$ for all non-queried databases and the average $\ell_p$ error would be a constant.
	If instead it set $f(x) = 1$ for all queried databases and $f(y) = 2$ for all non-queried databases, then the average $\ell_p$ error would approach 0 because the non-queried databases are infinite and the queried databases are finite. As a result, no finitely computable algorithm can achieve a constant approximation in this metric.
\end{remark}

%\begin{lemma}\label{lem:finite_approx_uncomp}
%	It is uncomputable to achieve a finite approximation for computing a $g$ with sensitivity at most $\Delta$ that minimizes the average $|f(D) - g(D)|$ over all databases for arbitrary $f$ and data universes $\mathcal{D}$
%\end{lemma}
%
%\begin{proof}
%	Similar to above
%	
%\end{proof}

% !TEX root = main.tex

\section{Efficient Implementation of Several Statistical Measures}\label{sec:efficient_examples}

%\todo{discuss using a single $\Delta$}

In this section, we take our general recursive algorithm and show how it can be made efficient for a variety of important statistical measures such as mean, $\alpha$-trimmed mean, median, minimum, and maximum. 
It is important to note that we will not change the key recursive structure, but instead show that when we have more information about the function, we can ignore many of the subproblems of the recursion for significant runtime speedups. 
As a result, the algorithm given for these statistical tasks will take $O(n^2)$ time and have a simple dynamic programming construction.

%Each of the functions considered in this section will map from the set of all finite-length real-valued vectors to the reals, or more explicitly $f:\R^{<\mathbb{N}} \rightarrow \R$. 
The key idea will be that given a database $D = (x_1,...,x_n)$ where we assume for simplicity that $x_1 \leq \cdots \leq x_n$,\footnote{Our algorithm will presort and only incur $O(n\log n)$ running time. } the only important subproblems will be $D - x_1$ and $D - x_n$. Consequently, instead of considering every possible subset of $D$, we only need to consider every contiguous subset, which limits the number of subproblems to $O(n^2)$.

%\rc{Add extensions to median, alpha trimmed mean, min, max, etc.}

We first give a general class of functions---which includes mean, median, $\alpha$-trimmed mean, minimum, and maximum---for which it is straightforward to show our algorithm can be applied efficiently. We then give a more in-depth analysis of the error guarantees that correspond with this implementation for mean. These bounds will ultimately be quite intuitive, but the proofs will be more involved.

\subsection{Efficient implementation for a simple class of functions}

We will first define a class of functions under which database ordering is preserved for any subset, which allows us to presort the data according to this ordering and restrict the number of subproblems. 
Intuitively, it implies that for any database $D = (x_1,...,x_n)$ there is an ordering of the $x_1,...,x_n$ such that the extreme points in our recursion are determined by the databases that remove the maximum or the minimum. 
In particular, if we consider the mean function $\mu:\R^{<\mathbb{N}} \rightarrow \R$ then for any $D = (x_1,...,x_n)$ if we assume $x_1\leq \cdots \leq x_n$, then we know $\mu(D - x_n) \leq \mu(D - x_i)$ and $\mu(D - x_1) \geq \mu(D - x_i)$ for any $i$. 
This will ultimately imply that our upper and lower bounds on the allowable region for $g(D)$ will be defined by $g(D - x_1)$ and $g(D - x_n)$, respectively.

\begin{definition}[Database-ordered function]
A function $f: \mathcal{D} \rightarrow \R$ is \textit{database-ordered} if for any $D = (x_1,...,x_n) \in \mathcal{D}$ and any pair $x_i,x_j \in D$, we have that for every subset database $D' \subset D$ such that $x_i,x_j \notin D'$, then either $f(D' + x_i) \leq f(D' + x_j)$ for every $D'$ or $f(D' + x_i) \geq f(D' + x_j)$ for every $D'$. Furthermore, if $f(D' + x_i) \leq f(D' + x_j)$ for every $D'$, we say that $x_i \leq x_j$ in the \textit{entry-ordering}, and vice-versa if $f(D' + x_i) \geq f(D' + x_j)$ for every $D'$.
	
\end{definition}

The general idea of our efficient implementation will be to use the ordering and only consider contiguous subsets according to this ordering.

\begin{lemma}\label{lem:database_ordered_is_efficient_implementable}
Given a database-ordered function $f: \mathcal{D} \rightarrow \R$, let $g: \mD \to \R$ be the \sensitivity~of $f$ with parameter $\Delta$. Then for any $D= (x_1,...,x_n)$ where $x_1 \leq \cdots \leq x_n$ in the \textit{entry-ordering} we must have $\upper(D) = g(D - x_n) + \Delta$ and $\lowerbound(D) = g(D - x_1) - \Delta$, and our \fitting~algorithm only requires solving $O(n^2)$ subproblems 
\end{lemma}

\begin{proof}
	We first want to show $\upper(D) = g(D - x_n) + \Delta$ and $\lowerbound(D) = g(D - x_1) - \Delta$. It is sufficient to show $g(D - x_n) \leq g(D - x_{n-1}) \leq \cdots \leq g(D - x_1)$, which we will prove by induction on the size of the database. 
	If $D$ only has one entry, then this must be true.
	
	Assume this is true for all $D$ with at most $n-1$ entries, and we want to show $g(D - x_{i+1}) \leq g(D - x_{i})$ for any $i \in [n-1]$. Since $f$ is database-ordered, we know that $f(D - x_{i+1}) \leq f(D - x_i)$. 
	It then suffices to show $\upper(D - x_{i+1}) \leq \upper(D - x_i)$ and $\lowerbound(D - x_{i+1}) \leq \lowerbound(D - x_i)$. By our inductive hypothesis, $\upper(D - x_i) = g(D - x_i - x_n) + \Delta$ and $\upper(D - x_{i+1}) = g(D - x_{i+1} - x_n) + \Delta$ if $i < n-1$, and we note that $\upper(D - x_{n-1}) = \upper(D - x_n)$. Also by our inductive hypothesis, $g(D - x_{i+1} - x_n) \leq g(D - x_i - x_n) $, implying $\upper(D - x_{i+1}) \leq \upper(D - x_i)$. The proof for $\lowerbound(D - x_{i+1}) \leq \lowerbound(D - x_i)$ follows symmetrically.
	
	With this fact, it is straightforward to see that opening up our algorithm, instead of considering all subsets of size $k$, it suffices to consider subsets $(x_1,\ldots, x_k), (x_2,\ldots, x_{k+1}),\ldots,(x_{n-k},\ldots,x_n)$. Then the total number of subproblems that need to be solved is $O(n^2)$.
\end{proof}

If our function is efficiently computable and the entry-ordering is efficiently computable, this then gives an efficient implementation of our recursive algorithm. In particular, for several functions of statistical interest including mean, $\alpha$-trimmed mean, median, maximum, and minimum, this easily yields an efficient algorithm.

\begin{algorithm}[h!]
	\caption{Efficient Implementation for database-ordered functions}
	\begin{algorithmic}
		\State \textbf{Input:} Database-ordered function $f:\R^{<\mathbb{N}} \to \bR$, sensitivity bound $\Delta$, estimate for the empty set $\hmu$, and database $D = (x_1,...,x_n) \in \R^n$ for some arbitrary $n$.
		\State \textbf{Output:} $g(D)$, where $g$ is the \sensitivity~of $f$.
		\State Initialize $g(\emptyset) = \hmu$
		\State Sort $D$ (We will assume $x_1 \leq \cdots \leq x_n$ for simplicity)
		\For {k=1, \ldots, n}
		\For {i = 1, \ldots n-k+1}
		\For {every database $D' = (x_i,...,x_{i+k-1})$}
		\State Let $g(D')= 
		\begin{cases}
		g(D' - x_{i+k-1}) + \Delta,& \text{if } g(D' - x_{i+k-1}) + \Delta \leq f(D')\\
		g(D' - x_{i}) - \Delta,& \text{if } g(D' - x_{i}) - \Delta\geq f(D')\\
		f(D'),              & \text{otherwise}
		\end{cases}$
		\EndFor
		\EndFor
		\EndFor
		\State Output $g(D)$		
	\end{algorithmic}\label{algo.mean}
\end{algorithm}

\begin{corollary}\label{cor:efficient_for_mean_etc}
	We can implement our \sensitivity~with parameter $\Delta$ in $O(n^2)$ time for the functions mean, $\alpha$-trimmed mean, median, maximum, and minimum.
\end{corollary}

\begin{proof}
	Let $f$ be any of the functions listed above. It is simple to see that for any $D = (x_1,...,x_n) \in \R^n$, and any $y,z \in \R$, if $y \leq z$ then $f(D + y) \leq f(D + z)$, and if $y \geq z$ then $f(D + y) \geq f(D + z)$. This implies that $f$ is database-ordered, then by Lemma~\ref{lem:database_ordered_is_efficient_implementable} we only need to solve $O(n^2)$ subproblems.
	
	Further, we note that finding the entry-ordering simply requires sorting the entries of $D$ in $O(n\log n)$ time. If the database is ordered, then computing median, minimum, and maximum only requires $O(1)$ time. If we know the mean or $\alpha$-trimmed mean for $D - x_i$ for some $x_i$, we can compute the mean or $\alpha$-trimmed mean of $D$ in $O(1)$ time using the fact that
	\[
	\frac{x_1+...+x_n}{n} = \frac{n-1}{n}\left(\frac{x_1+...+x_{n-1}}{n-1}\right) + \frac{x_n}{n}
	\]
	Note that we compute $D- x_i$ for some $i$ in our subproblems, so we will in fact have access to this value.
	As a result, the full running time will take $O(n^2)$ time.
	
\end{proof}

\subsection{Improved runtime and accuracy for median}\label{subsec:median}

In the previous section, we showed that for several important statistical measures we could give a simple efficient version of our general algorithm. 
To complement this result, we further examine the median function and give an improved analysis that requires only $O(n)$ time for presorted data and provides strong accuracy guarantees.
Improving the running time will utilize the critical property that removing the minimum and maximum value does not change the median.
As was seen in our previous section, our recursion was reduced by only considering removing the maximum or minimum value.
The related fact regarding median will be incorporated into an inductive claim that we never overshoot the true median, and can further reduce our recursion.
%While the analysis will be rather involved, we believe that the ultimate guarantees are highly intuitive.
%Our proof will also show that for databases with entries bounded in a $\Delta$ sensitivity range, we perfectly preserve the accuracy between our new function and the mean function.
%Further, the key ideas in our proof are closely related to the construction of our recursive function, and we believe could be extended to other functions using a similar framework.

%\rc{write more about how the proof works}

\begin{lemma}\label{lem:med_runtime}
Let $med: \R^{<\mathbb{N}} \rightarrow \R$ be the median function and $g: \R^{<\mathbb{N}} \to \R$ be the \sensitivity~of $med$ with parameter $\Delta$. Then for any $D= (x_1,...,x_n)$ such that $x_1 \leq \cdots \leq x_n$, computing $g(D)$ takes $O(n)$ time.

\end{lemma}

\begin{proof}
It follows immediately from Lemma~\ref{lem:database_ordered_is_efficient_implementable} and Lemma~\ref{lem:median_above_and_below} that if $med(D) \geq med(\emptyset)$ then $g(D) = \min\{med(D), g(D - x_n) + \Delta \}$ and otherwise $g(D) = \max\{med(D), g(D - x_1) - \Delta$ \}. We can calculate $med(D)$ and any contiguous subset of $D$ in $O(1)$ time, and the recursion will only be upon one subproblem, implying a runtime of $O(n)$.
\end{proof}

\begin{lemma}\label{lem:median_above_and_below}
If $med(D) \geq med(\emptyset)$, then $med(\emptyset) \leq g(D) \leq med(D)$
\end{lemma}

\begin{proof}
The proof will be inductive, and it is easy to verify that the inequality holds for $|D| \leq 2$. We then consider an arbitrary $D = (x_1,...,x_n)$ where we assume without loss of generality that $x_1 \leq \cdots \leq x_n$ and $n \geq 3$. The critical fact we use here will be that the median does not change if you remove the minimum and maximum values, which is to say that $med(D) = med(D - x_1 - x_n)$. Therefore, if $med(D) \geq med(\emptyset)$, then we must also have $med(D - x_1 - x_n) \geq med(\emptyset)$, which by our inductive claim implies that $med(\emptyset) \leq g(D - x_1 - x_n) \leq med(D - x_1 - x_n) = med(D)$. Applying Lemma~\ref{lem:database_ordered_is_efficient_implementable}, we then have 

\[
g(D - x_1) \leq g(D - x_1 - x_n) + \Delta \leq med(D) + \Delta
\]
and
\[
g(D - x_n) \geq g(D - x_1 - x_n) - \Delta \geq med(D) - \Delta
\]

We then reapply Lemma~\ref{lem:database_ordered_is_efficient_implementable} to achieve our desired result that $med(\emptyset) \leq g(D) \leq med(D)$
\end{proof}

%As in the smooth sensitivity paper, define $A^{(k)}(D) = \max_{0 \leq t \leq k+1} (x_{m + t} - x_{m + t - k - 1})$, and $m = \frac{n+1}{2}$ where this is essentially the $k$-local sensitivity for median for database $D$.

As in \cite{NRS07}, define 
\[
A^{(k)}(D) = \max_{d(D,D') \leq k} LS_f(D').
\]
which is the $k$-local sensitivity of function $f$ for database $D$. For odd $n$, this just reduces to 
$A^{(k)}(D) = \max_{0 \leq t \leq k+1} (x_{m + t} - x_{m + t - k - 1})$ and $m = \frac{n+1}{2}$. It is similar for $n$ is even, and essentially bounds the distance of each value from the median.

Combining this assumption with our previous lemma will then allow for stronger bounds upon $g(D)$.

\begin{lemma}
Given some parameter $\Delta$ and $med(\emptyset)$, if $A^{(k)}(D) \leq 2(k + 1)\Delta$ for $k \leq n/4$ and $med(D) \in [med(\emptyset) - \frac{n}{2}\Delta, med(\emptyset) + \frac{n}{2}\Delta]$, then $g(D) = med(D)$
\end{lemma}

\begin{proof}
Without loss of generality, assume that $med(D) \geq med(\emptyset)$. By Lemma~\ref{lem:median_above_and_below} we know $g(D) \leq med(D)$, then applying Lemma~\ref{lem:median_accuracy_helper} gives our desired result.
\end{proof}

\begin{lemma}\label{lem:median_accuracy_helper}
Given some parameter $\Delta$ and $med(\emptyset)$, assume $A^{(k)}(D) \leq 2(k + 1)\Delta$ for $k \leq n/4$ and $med(D) \in [med(\emptyset) - \frac{n}{2}\Delta, med(\emptyset) + \frac{n}{2}\Delta]$. Let $D_{[1:k]} = (x_1,...,x_k)$, if $med(D) \geq med(\emptyset)$, then $g(D_{[1:k]}) \geq med(D) - (n - k) \Delta$
\end{lemma}

\begin{proof}
It is straightforward to see that our assumptions imply

\[
med(D_{[1:k]}) \geq med(D) - (n - k)\Delta
\]
for any $k \geq n/2$. We then consider our base case to be $k = n/2$, and note that from Lemma~\ref{lem:median_above_and_below} we have $g(D_{[1:k]}) \geq \min \{med(\emptyset), med(D_{[1:k]}) \}$, which by our assumptions immediately implies $g(D_{[1:n/2]}) \geq med(D) - \frac{n}{2}\Delta$.

We then assume this is true for $k - 1 \geq n/2$, so $g(D_{[1:k-1]}) \geq med(D) - (n-k)\Delta - \Delta$. We also know from Lemma~\ref{lem:database_ordered_is_efficient_implementable} that 

\[
g(D_{[1:k]}) \geq \min \{med(D_{[1:k]}), g(D_{[1:k-1]}) + \Delta \}
\]

which implies our desired inequality.

\end{proof}

\subsubsection{Proof of Theorem~\ref{thm:median}}

We now have all the necessary components to give our proof of Theorem~\ref{thm:median}, which we restate and prove below.

\median*

\begin{proof}[Proof of Theorem~\ref{thm:median}] 
The runtime guarantees follow immediately from Lemma~\ref{lem:med_runtime}. Furthermore, if we assume that $med(D) \geq med(\emptyset)$, then Lemma~\ref{lem:median_above_and_below} implies that $g(D) \leq med(D)$ and Lemma~\ref{lem:median_accuracy_helper} implies that $g(D) \geq med(D)$ because we have the same assumptions, and so $g(D) = med(D)$
The symmetric version of these lemmas follows immediately, and we also have $g(D) = med(D)$ when $med(D) \leq med(\emptyset)$.

\end{proof}

\subsection{Accuracy bounds for mean}\label{s.meanacc}

%In the previous section, we showed that for several important statistical measures we could give a simple efficient version of our general algorithm. 
%To complement this result, we further examine the mean function and give strong bounds on the accuracy guarantees. 
We next consider the mean function, and provide strong bounds on the accuracy of our \sensitivity.
While the analysis will be rather involved, we believe that the ultimate guarantees are highly intuitive.
Our proof will also show that for databases with entries bounded in a $\Delta$ sensitivity range, we perfectly preserve the accuracy between our new function and the mean function.
Further, the key ideas in our proof are closely related to the construction of our recursive function, and we believe could be extended to other functions using a similar framework.

The general proof idea will be to give two simpler recursive functions that yield reasonably tight upper and lower bounds on our function.
Due to their further simplicity, it will be much easier to give nice error bounds with respect to the true mean for these functions.

The idea behind constructing the upper and lower bound functions will be simple. 
Recall that we showed our $g$ for the mean function has the property that $\upper(D) = g(D - x_n) + \Delta$ and $\lowerbound(D) = g(D - x_1) - \Delta$ because we showed $g(D - x_n) \leq g(D - x_{n-1}) \leq \cdots \leq g(D - x_1)$ if we assume $x_1 \leq \cdots \leq x_n$.
Intuitively, this is due to the fact that removing the maximum value will minimize mean and removing the minimum value will maximize mean. 
Accordingly, we will just iteratively remove the maximum value to give a lower bound on our function and iteratively remove the minimum value to give an upper bound on our function.
These functions then only require solving $O(n)$ subproblems which will simplify the analysis.

\begin{definition}[Mean-bounding functions]
	For any $D = (x_1,...,x_n)$, define 
	\[
	h_{lower}(D)= 
	\begin{cases}
	h_{lower}(D - x_n) + \Delta,& \text{if } h_{lower}(D - x_n) + \Delta\leq \mu(D)\\
	h_{lower}(D - x_n) - \Delta,& \text{if } h_{lower}(D - x_n) - \Delta\geq \mu(D)\\
	\mu(D),              & \text{otherwise}
	\end{cases}
	\] 
	and
	\[
	h_{upper}(D)= 
	\begin{cases}
	h_{upper}(D -  x_1) + \Delta,& \text{if } h_{upper}(D -  x_1) + \Delta\leq \mu(D)\\
	h_{upper}(D -  x_1) - \Delta,& \text{if } h_{upper}(D -  x_1) - \Delta\geq \mu(D)\\
	\mu(D),              & \text{otherwise}
	\end{cases}
	\] 
\end{definition}

We will first show that $h_{upper}$ and $h_{lower}$ are upper and lower bounds, respectively, of our \sensitivity~$g$ with parameter $\Delta$. Then we further examine the properties of these functions.

\begin{lemma}\label{lem:upper_and_lower_bound_g}
	Let $\mu: \R^{<\mathbb{N}} \rightarrow \R$ be the mean function with chosen parameters $\hmu$ and $\Delta$. For any $D = (x_1,...,x_n) \in \R^n$ with $x_1 \leq x_2 \leq \cdots \leq x_n$, then $h_{lower}(D) \leq g(D)$  and $h_{upper}(D) \geq g(D)$ where $g: \R^{<\mathbb{N}} \rightarrow \R$ is our \sensitivity~with parameter $\Delta$.
\end{lemma}

\begin{proof}
	We will prove both inequalities by induction, where we first note that if $D$ only has one entry, then by construction $h_{lower}(D) = g(D) = h_{upper}(D)$.
	
	For any database $D$ of $n$ entries, by induction we have $h_{lower}(D - x_n) \leq g(D - x_n)$ and note that within the proof of Lemma~\ref{lem:database_ordered_is_efficient_implementable} we showed $g(D - x_n) \leq g(D - x_1)$, which implies $h_{lower}(D) \leq g(D)$. Similarly, by induction we have $h_{upper}(D - x_1) \geq g(D - x_1)$ and Lemma~\ref{lem:database_ordered_is_efficient_implementable} gives $g(D - x_1) \geq g(D - x_n)$, which implies $h_{upper}(D) \leq g(D)$.	
\end{proof}

%\david{notation is shit here}

We now use the simpler recursive structure of $h_{lower}$ and $h_{upper}$ to get more explicit forms of their output.

\begin{lemma}\label{lem:mean_first_index_for_accuracy}
	Let $\mu: \R^{<\mathbb{N}} \rightarrow \R$ be the mean function with chosen parameters $\hmu$ and $\Delta$.
	For any $D = (x_1,...,x_n)$, assume that $x_1 \leq x_2 \leq \cdots \leq x_n$, and let $D_{[i:j]} = (x_i,...,x_j)$. Let $k$ be the largest index such that $h_{lower}(D_{[1:k]}) \geq \mu({D_{[1:k]}})$ (if one exists), then
	\[
	h_{lower}(D_{[1:k]}) = \max \{\hmu - k\Delta, \mu({D{[1:k]}})  \}.
	\]
	Let $l$ be the smallest index such that $h_{upper}(D_{[l:n]}) \geq \mu({D_{[l:n]}})$ (if one exists), then
	\[
	h_{upper}(D_{[l:n]}) = \min \{\hmu + (n-l)\Delta, \mu({D_{[l:n]}})  \}.
	\]
\end{lemma}

%\todo{make $\mu_D$ notation always match}
\begin{proof}
	We consider the first equality here, and the second follows symmetrically. 
	
	Note that $\mu({D_{[1:k]}})$ is increasing in $k$ because $x_1 \leq \cdots \leq x_n$. By construction of $h_{lower}$, if for some index $k'$ we have $h_{lower}(D_{[1:k']}) \leq \mu({D_{[1:k']}})$, then $h_{lower}(D_{[1:k'+1]}) \leq \mu({D_{[1:k'+1]}})$. 
	Accordingly, if we let $k_{min}$ be the first index such that $h_{lower}(D_{[1,k_{min}]}) \leq \mu(D_{1,k_{min}})$, then in the case that $k \geq k_{min}$ we must have $h_{lower}(D_{[1:k]}) = \mu(D_{[1:k]})$. If $k < k_{min}$, then we must have $h_{lower}(D_{[1:k]}) > \mu(D_{[1:k]})$, and furthermore $h_{lower}(D_{[1:k']}) > \mu(D_{[1:k']})$ for all $k' \leq k$, which implies that we always decreased by $\Delta$ and we get $h_{lower}(D_{[1:k]}) = \hmu - k\Delta$.	
\end{proof}

We use the explicit forms of $h_{lower}$ and $h_{upper}$ to sandwich the loss in accuracy, by considering the inflection point of $n/3$ and bounding the error from $h_{lower}$ separately for $k \leq n/3$ and for $k \geq n/3$. 
The analogous result follows symmetrically for $h_{upper}$.

\begin{lemma}\label{lem:mean_error_bounds}
	Let $\mu: \R^{<\mathbb{N}} \rightarrow \R$ be the mean function with chosen parameters $\hmu$ and $\Delta$. If $g: \R^{<\mathbb{N}} \rightarrow \R$ is our \sensitivity~with parameter $\Delta$, then given any $D = (x_1,...,x_n)$,
	\[
	\abs{g(D) - \mu(D)} \leq \max\{\abs{\hmu - \mu(D)} - \frac{n}{3} \Delta, 0 \} + \sum_{i=1}^n \max\left\{\frac{27\abs{x_i - \mu(D)}}{n} - \Delta,0\right\}.
	\]
\end{lemma}

\begin{proof}
	If we can instead prove the same upper bounds for both $\abs{h_{lower}(D) - \mu(D)}$ and $\abs{h_{upper}(D) - \mu(D)}$, then the desired bound for $\abs{g(D) - \mu(D)}$ follows from Lemma~\ref{lem:upper_and_lower_bound_g}. We give the desired bound for $\abs{h_{lower}(D) - \mu(D)}$, and the bound for $\abs{h_{upper}(D) - \mu(D)}$ follows symmetrically.
	
	Again, let $k$ be the largest index such that $h_{lower}(D_{[1:k]}) \geq \mu({D_{[1:k]}})$ (if one exists). If $k \leq n/3$ or none exists, then it immediately follows from Lemma~\ref{lem:mean_first_index_for_accuracy} that $h_{lower}(D) \geq \hmu + \frac{n}{3}\Delta$, which implies $|\mu(D) - h_{lower}(D)| \leq |\mu(D) - \hmu| - \frac{n}{3}\Delta$.
	
	If $k \geq n/3$, then it is implied by Lemma~\ref{lem:mean_first_index_for_accuracy} that $h_{lower}(D)  = \max \{\hmu - k\Delta, \mu({D_{[1:k]}})  \} + (n-k)\Delta \geq \mu({D_{[1:k]}}) + (n-k)\Delta $ and therefore,
	\[
	\mu_D - h_{lower}(D) \leq \mu(D) - \mu({D_{[1:k]}}) + (n-k)\Delta = \sum_{i=k}^{n-1} \left(\mu({D_{[1:i+1]}}) - \mu({D_{[1:i]}}) \right) - (n - k)\Delta.
	\]
	Furthermore, 
	\[
	\mu({D_{[1:i+1]}}) - \mu({D_{[1:i]}})= \frac{x_1 + \cdots + x_{i+1}}{i+1} - \frac{x_1 + \cdots + x_i}{i} = \frac{1}{i(i+1)}\left( \sum_{j=1}^i x_{i+1} - x_j \right).
	\]
	
%	which separates out to 
%	
%	\[
%	\frac{x_{i+1}}{i+1} - \frac{x_1 + \cdots + x_i}{i(i+1)} - \Delta = \frac{1}{i(i+1)}\left( \sum_{j=1}^i x_{i+1} - x_j - (i+1)\Delta \right)
%	\]

	We use the fact that $i \geq n/3$ to achieve,
	\[
	\mu(D) - h_{lower}(D) \leq \left(\frac{9}{n^2} \sum_{i=k}^n\sum_{j=1}^i (x_{i} - x_j) \right)- (n - k)\Delta.
	\]
	
	Applying Lemma~\ref{cor:bound_on_sum_of_absolute_values} (stated below) gives, 
	\[
	\mu(D) - h_{lower}(D) \leq \left(\frac{27}{n} \sum_{i=k}^n \abs{x_{i} - \mu(D)} \right)- (n - k)\Delta =  \sum_{i=k}^n \left( \frac{27\abs{x_{i} - \mu(D)}}{n} -  \Delta \right)
	\]
	We then add in non-negative terms that are necessary for the symmetric version with $h_{upper}$ to achieve our desired bound.
\end{proof}

We used the following lemma to simplify the bounds in Lemma \ref{lem:mean_error_bounds} beyond those stated in the more general Lemma~\ref{lem:1D_error_bounds}. We relegate the proof of this lemma to the appendix.

\begin{lemma}\label{cor:bound_on_sum_of_absolute_values}
	For any set of reals $D = (x_1,...,x_n)$ where $x_1 \leq \cdots \leq x_n$, given any index $k \in [n]$,
	\[
	\frac{1}{n^2} \sum_{i=k}^n \sum_{j=1}^i \frac{1}{3} \abs{x_i - x_j} 
	\leq \frac{1}{n} \sum_{i=k}^n \abs{x_i - \mu(D)}.
	\]	
\end{lemma} 

To finally obtain all the necessary components for the proof of Theorem~\ref{thm:mean}, it is only left to show that when all the inputs of the database are in a nicely bounded range, our \sensitivity~will perfectly fit to the function $\mu$.

%This next lemma is nice because if we consider all databases of $n$ entries such that each entry is in some range of $n\Delta$, then the sensitivity is $\Delta$. It would then be nice if we similarly got perfect accuracy for this restricted set of databases, which we do (so our algorithm does "find" the good databases).

\begin{lemma}\label{lem:mean_does_good_on_wellbehaved}
	Let $\mu: \R^{<\mathbb{N}} \rightarrow \R$ be the mean function with chosen parameters $\hmu$ and $\Delta$. If $g: \R^{<\mathbb{N}} \rightarrow \R$ is our \sensitivity~with parameter $\Delta$, then given any $D = (x_1,...,x_n) \in \R^n$, if for all $x_i \in D$ we have $x_i \in [\hmu + \alpha \Delta, \hmu + (\alpha + n)\Delta]$ for $\alpha \in [-n,0]$, then $g(D) = \mu(D)$.
\end{lemma}

\begin{proof}
	First, it is straightforward to see by the construction of $h_{lower}$ and $h_{upper}$ that $h_{lower}(D) \leq \mu(D)$ if $\mu(D) \geq \hmu - n \Delta$ and $h_{upper}(D) \geq \mu(D)$ if $\mu(D) \leq \hmu + n \Delta$. 
	Therefore, by Lemma~\ref{lem:upper_and_lower_bound_g}, the desired result is implied if $h_{lower}(D) \geq \mu(D)$ and $h_{upper}(D) \leq \mu(D)$.
	Here we show that $h_{lower}(D) \geq \mu(D)$, and $h_{upper}(D) \leq \mu(D)$ will be implied symmetrically.
	
	%	Once again, let $k$ be the largest index such that $h_{lower}(D_{[1:k]}) \geq \mu_{D_{[1:k]}}$, which implies by Lemma~\ref{lem:mean_first_index_for_accuracy} that $h_{lower}(D) = h_{lower}(D_{[1:k]}) + (n-k)\Delta$ and $h_{lower(D)} \leq \mu_D$. Therefore, it suffices to show 
	
	Suppose it is not true that $h_{lower}(D) \geq \mu(D)$, then there must exist the last index $k < n$ such that $h_{lower}(D_{[1:k]}) \geq \mu({D_{[1:k]}})$, which by construction implies that $h_{lower}(D) = h_{lower}(D_{[1:k]}) + (n-k)\Delta$. To achieve our contradiction, we want to show that $\mu(D) - h_{lower}(D_{[1:k]}) \leq (n-k) \Delta$.
	
	By our restriction of each $x_i$ and by assumption we have,
	\[
	\hmu + \alpha \Delta \leq \mu({D_{[1:k]}}) \leq h_{lower}(D_{[1:k]}).
	\]
	Furthermore, because all of the remaining $x_i \leq \hmu + (\alpha + n)\Delta$, we must have, 
	\[
	\mu(D) \leq \frac{k \mu({D_{[1:k]}}) + (n-k) (\hmu + (\alpha + n)\Delta)}{n} \leq 
	\frac{k \cdot h_{lower}(D_{[1:k]}) + (n-k) (\hmu + (\alpha + n)\Delta)}{n},
	\]
	where the second inequality follows from our assumption that $h_{lower}(D_{[1:k]}) \geq \mu({D_{[1:k]}}) $. This implies, 
	\begin{align*}
	\mu(D) - h_{lower}(D_{[1:k]}) &\leq \frac{k \cdot  h_{lower}(D_{[1:k]}) + (n-k) (\hmu + (\alpha + n)\Delta)}{n} - h_{lower}(D_{[1:k]})
	\\ 
	&= \frac{(k-n) h_{lower}(D_{[1:k]}) + (n-k) (\hmu + \frac{n}{2}\Delta)}{n}
	\end{align*}
	
	We use the fact that $h_{lower}(D_{[1:k]}) \geq \hmu + \alpha \Delta$ and $k < n$ to get,
	\[
	\mu(D) - h_{lower}(D_{[1:k]}) \leq \frac{(k - n) (\hmu + \alpha\Delta) + (n-k) (\hmu + (\alpha + n)\Delta)}{n} = (n-k) \Delta,
	\]
	giving our desired contradiction, which implies $h_{lower}(D) \geq \mu(D)$.
\end{proof}

\subsubsection{Proof of Theorem~\ref{thm:mean}}

We now have all the necessary components to give our proof of Theorem~\ref{thm:mean}, which we restate and prove below.

\mean*

\begin{proof}[Proof of Theorem~\ref{thm:mean}]
	The fact that $g$ has sensitivity $\Delta$ follows from the fact that it is our \sensitivity~and the guarantees of Lemma~\ref{lem:1d_sensitivity_guarantees}.
	The runtime follows from Corollary~\ref{cor:efficient_for_mean_etc}.
	We then achieve the error bounds from Lemma~\ref{lem:mean_error_bounds} and Lemma~\ref{lem:mean_does_good_on_wellbehaved}.

\end{proof}

% !TEX root = main.tex

\section{Efficient Implementation for Variance}\label{sec:variance_algo}

%\todo{Justify this by efficient implementation of a hard function}

In this section, we show how to efficiently extend our recursive algorithm to \emph{variance}, which is an important statistical metric and a more complicated function than those considered in Section \ref{sec:efficient_examples}. Although variance is not a database-ordered function, we can still implement our \sensitivity~for variance in $O(n^2)$ time, using similar techniques to reduce the number of subproblems that must be considered.  This suggests that database-ordered functions are not the only class that have an efficient implementation, and that running time of our algorithm can be improved more generally using structural properties of the function being considered.

%In this section, we consider a more complicated function, variance, that is also an important statistical metric and show how to efficiently extend our recursive algorithm. 
The general idea will remain the same as we reduce the number of subproblems to $O(n^2)$ by using structural properties of variance. 
We first formally define the discrete version of variance with two equivalent equations.

\begin{definition}
	For any $D = (x_1,...,x_n) \in \R^n$, let $\mu(D) = \frac{1}{n}(x_1 + \cdots + x_n)$ and define the variance function,
	\[
	\var{D} \defeq \frac{1}{n} \sum_{i=1}^n \left(x_i - \mu(D)\right)^2,
	\]
or equivalently,
	\[
	\var{D} \defeq \frac{1}{n^2} \sum_{i=1}^n \sum_{j=1}^n \frac{1}{2} \left(x_i - x_j\right)^2.
	\]
\end{definition}

As with mean, $\alpha$-trimmed mean, median, maximum, and minimum, we will first sort the entries of the database. 
Intuitively, we can decrease the variance most by removing either the minimum or maximum value.  We make use of the following fact, which we prove in the appendix for completeness.

\begin{fact}\label{fact:variance_minimized_with_min_max_tossed}
	Given $D = (x_1,...,x_n) \in \R^n$ such that $x_1 \leq \cdots \leq x_n$, then for any $i$,
	\[
	\min \{\var{D-x_1},\var{D - x_n} \} \leq \var{D - x_i}.
	\]
\end{fact}

We will use this fact to show that the lower bound on $g(D)$ will be defined by $g(D - x_1)$ or $g(D - x_n)$. 
The difficulty now becomes that to increase variance the most, we would want to remove an entry between $x_1$ and $x_n$. 
This poses a significant complication in constructing a dynamic program for the subproblems.
More specifically, even if $g(D)$ only required solving two subproblems $g(D - x_i)$ and $g(D - x_j)$ for some $x_i,x_j$, we are still doubling the number of subproblems at each step.
The straightforward dynamic program for ordered-databases was able to reuse different subproblems to avoid a runtime blow-up.
The key idea will then be that we can bound, with respect to the original variance, the amount variance can be increase by removing an entry.
In particular, we use the following fact that is likely a folklore result, but we could not find a citation, so we prove it in the appendix for completeness.

\begin{fact}\label{fact:bound_marginal_variance}
	Given any unordered $(x_1,...,x_n) \in \R^n$,
	\[
	\var{x_1,...,x_{n-1}} \leq \frac{n}{n-1} \var{x_1,...,x_n}.
	\]	
\end{fact}

We can then use this strong bound to show that if we initialize $g(\emptyset) = 0$, the \sensitivity~will never go above $\var{D}$ for any $g(D)$.
As a result, the \sensitivity~will never actually use $\lowerbound(D)$.
This will then allow us to only recurse on subproblems where the minimum or maximum has been removed, and the dynamic program will be analogous to the one given for mean.

We first give the efficient implementation for variance and show that it can be done in $O(n^2)$ time. Then we give stronger bounds on the error incurred by this efficient implementation, and finally use these facts to prove Theorem~\ref{thm:variance}.

\subsection{Efficient algorithm for variance}

As with mean and the database-ordered functions, the key to our efficient implementation will be showing that the \sensitivity~can be equivalently defined using far fewer subproblems.
Using some of the intuition above, we are able to prove the following lemma that reduces the \sensitivity~to a much simpler recursion.

\begin{lemma}\label{lem:variance_efficient_recursion}
	Let $\textbf{Var}:\R^{<\mathbb{N}} \to \bR$ be the variance function and set $\var{\emptyset} = 0$. Then the \sensitivity~with parameter $\Delta$ can be equivalently defined as $g(\emptyset) = 0$ and $g(D) = \min\{\var{D}, g(D - x_1) + \Delta, g(D - x_{n}) + \Delta\}$ where $D = (x_1,...,x_n)$ with $x_1 \leq \cdots \leq x_n$.
\end{lemma}

We will prove this lemma with the following two helper lemmas. The first will show that the \sensitivity~will never exceed the true variance. The second uses the fact that variance is minimized by either removing the minimum or maximum value to show that the lower bound can simply consider the subproblems $g(D - x_1)$ and $g(D - x_n)$.

\begin{lemma}\label{lem:variance_doesnt_go_above}
	Given any $D = (x_1,...,x_n) \in \R^n$, if $g$ is the \sensitivity~of variance with parameter $\Delta$ and $g(\emptyset) = 0$, then,
	\[
	g(D) \leq \var{D}.
	\]
\end{lemma}

\begin{proof}
	We will prove this by induction. If $D$ contains only a single entry, then $\var{D} = 0$ and by construction $g(D) = 0$. 
	
	We then consider $D = (x_1,...,x_n)$ and assume the inequality holds for all subsets. By the definition of the \sensitivity, it suffices to show that $g(D - x_i) - \Delta \leq \var{D}$ for all $x_i$. Our inductive claim gives $g(D - x_i) \leq \var{D - x_i}$, and Fact~\ref{fact:bound_marginal_variance} implies:
	\[
	\var{D-x_i} - \var{D} \leq \frac{1}{n-1}\var{D},
	\]
These combine to give,
	\[
	g(D - x_i) - \var{D} \leq \frac{1}{n-1}\var{D}.
	\]
	
	We now consider two cases. If $g(D - x_i) \leq \var{D}$, then $g(D - x_i)  - \Delta \leq \var{D}$ because $\Delta \geq 0$ and we have our desired inequality. If $\var{D} \leq g(D - x_i)$ then,
	\[
	g(D - x_i) - \var{D} \leq \frac{1}{n-1}g(D-x_i).
	\]	
	Further, by the definition of \sensitivity~and the fact that $g(\emptyset) = 0$, we must have $g(D - x_i) \leq (n-1)\Delta$, implying,
	\[
	g(D - x_i) - \var{D} \leq \Delta,
	\]
which is our desired inequality.
\end{proof}

%\begin{corollary}\label{cor:variance_doesnt_need_lower}
%	If the function we are looking at is $\var{D}$, then by setting $g(\emptyset) = 0$, we have
%	
%	\[
%	g(D) = \min \{\upper(D), \var{D} \}
%	\]	
%	
%\end{corollary}

\begin{lemma}\label{lem:upper_can_get_rid_of_minmax}
	Given $D = (x_1,...,x_n) \in \R^n$ such that $x_1 \leq \cdots \leq x_n$ and $g$ is the \sensitivity~of variance with parameter $\Delta$ and $g(\emptyset) = 0$, then,
	\[
	\min\{ g(D - x_1), g(D - x_n)\} \leq g(D - x_i),
	\]
	for any $x_i \in D$.	
\end{lemma}

\begin{proof}
	We will prove this by induction. If $D$ has just one entry then $x_1 = x_i = x_n$ and each term is equivalent.
	
We then consider $D = (x_1,...,x_n)$ and assume the inequality holds for all subsets. We will consider two cases. Our first case is  $g(D - x_i) = \var{D - x_i}$. Lemma~\ref{lem:variance_doesnt_go_above} implies:
	\[
	\min\{g(D - x_1),g(D - x_n)\} \leq \min \{\var{D - x_1},\var{D - x_n} \}.
	\]
	Furthermore, by Fact~\ref{fact:variance_minimized_with_min_max_tossed} we have $\min \{\var{D - x_1},\var{D - x_n} \} \leq \var{D - x_i}$.  Combining this with the assumption $g(D - x_i) = \var{D - x_i}$ gives the desired inequality.
	
	It is implied by Lemma~\ref{lem:variance_doesnt_go_above} that the only other case we need to consider is $g(D - x_i) < \var{D - x_i}$. This assumption and our definition of \sensitivity~together imply,
	\[
	g(D - x_i) = \min_{j \neq i} \{g(D - x_i - x_j) + \Delta\}.
	\]
	The definition of \sensitivity~also gives:
	\[
	\min\{g(D - x_1),g(D - x_n)\} \leq \min \{\min_{j\neq 1} \{g(D - x_1 - x_j) + \Delta \}, \min_{j\neq n} \{g(D - x_n - x_j) + \Delta \} \}.
	\]
	
	As a result, if $\min_{j\neq 1} \{g(D - x_i - x_j) + \Delta \}$ is minimized for $j=1$ or $j=n$, then we easily have $\min\{g(D - x_1),g(D - x_n)\} \leq g(D - x_i)$. Furthermore, if $j \neq 1,n$, then it suffices to show that,
	\[
	\min \{g(D - x_1 - x_j),g(D - x_n - x_j)\} \leq g(D - x_i - x_j),
	\]
	which follows from the inductive hypothesis and implies our desired result.
\end{proof}

These two helper lemmas now easily imply Lemma~\ref{lem:variance_efficient_recursion}.

\begin{proof}[Proof of Lemma~\ref{lem:variance_efficient_recursion}]
	Lemma~\ref{lem:variance_doesnt_go_above} implies that we will never need to use $\lowerbound(D)$, so we can eliminate that case. Further, Lemma~\ref{lem:upper_can_get_rid_of_minmax} implies that $\upper(D) = \min\{g(D - x_1) + \Delta, g(D - x_{n}) + \Delta \}$. Combining these facts implies our recursion defined in the lemma statement is equivalent to the \sensitivity.
\end{proof}

With this reduction in the number of subproblems for the \sensitivity, we will be able to give a similar efficient dynamic programming algorithm for the implementation. 

\begin{algorithm}[h!]
	\caption{Efficient Implementation for Variance}
	\begin{algorithmic}
		\State \textbf{Input:} Variance function $\textbf{Var}:\R^{<\mathbb{N}} \to \bR$, sensitivity bound $\Delta$, and database $D = (x_1,...,x_n) \in \R^n$ for some arbitrary $n$.
		\State \textbf{Output:} $g(D)$ where $g$ is the \sensitivity~of variance with parameter $\Delta$.
		\State Initialize $g(\emptyset) = 0$
		\State Sort $D$ (We will assume $x_1 \leq \cdots \leq x_n$ for simplicity)
		\For {k=1, \ldots, n}
		\For {i = 1, \ldots n-k+1}
		\For {every database $D' = (x_i,...,x_{i+k-1})$}
		\State Let $g(D')= \min\{\var{D'}, g(D' - x_i) + \Delta, g(D' - x_{i+k-1}) + \Delta\}$
		\EndFor
		\EndFor
		\EndFor
		\State Output $g(D)$		
	\end{algorithmic}\label{algo.variance}
\end{algorithm}

It immediately follows that the number of subproblems that we need to consider is $O(n^2)$, but we still need to efficiently compute $\var{D}$. This computation would normally take $O(n)$ time and increase our running time to $O(n^3)$. However, we can use the computation from previous subproblems to compute the variance in $O(1)$ time with the following folklore fact that we prove in the appendix.

\begin{fact}\label{fact:variance_in_constant_time}
	For any $D = (x_1,...,x_n) \in \R^n$ and any $x_a \neq x_b \in D$,
	\[
	\var{D} = \left(\frac{n-1}{n}\right)^2\var{D - x_a} + \left(\frac{n-1}{n}\right)^2\var{D - x_b} - \left(\frac{n-2}{n}\right)^2\var{D - x_a - x_b}
	+ \frac{1}{n^2}(x_a - x_b)^2.
	\]
\end{fact} 

With this fact we can now show that we implement the \sensitivity~for variance with parameter $\Delta$ in $O(n^2)$ time.

\begin{lemma}\label{lem:variance_runtime_and_correctness}
	Let $\textbf{Var}:\R^{<\mathbb{N}} \to \bR$ be the variance function and set $\var{\emptyset} = 0$. Then Algorithm~\ref{algo.variance} will compute $g(D)$ for any database of $n$ entries in $O(n^2)$ time where $g$ is the \sensitivity~for variance with parameter $\Delta$.
\end{lemma}

\begin{proof}
	Correctness of the procedure follows immediately from Lemma~\ref{lem:variance_efficient_recursion}. The running time follows from the fact that we have $O(n^2)$ subproblems and from Fact~\ref{fact:variance_in_constant_time} we can compute $\var{D}$ in $O(1)$ time using the previous subproblems.
\end{proof}

\subsection{Accuracy guarantees for variance implementation}

In this section we give stronger bounds on the error incurred by the \sensitivity. 
The proofs will be similar to those in Section \ref{s.meanacc} for mean, but will be slightly simpler due to that fact that the \sensitivity~will never go above the actual variance.
As a result, we achieve a simpler form for the error of the \sensitivity~with respect to variance in the following lemma.

\begin{lemma}\label{lem:simplification_of_variance_value}
	Given $D = (x_1,...,x_n) \in \R^n$ and $g$ that is the \sensitivity~of variance with parameter $\Delta$ and $g(\emptyset) = 0$, then there must exist some $D' \subseteq D$ such that $g(D) = \var{D'} + (n-k)\Delta$ for $k=|D'|$.
\end{lemma}

\begin{proof}
	We prove this inductively on the size of $D$ and see immediately that the claim holds by construction for $D$ with a single entry.
	
	We then consider $D = (x_1,...,x_n)$ and assume that our claim holds for all subsets.
	From Lemma~\ref{lem:variance_doesnt_go_above} we know that $\var{D} \geq g(D)$ for all databases. If $\var{D} = g(D)$, then our claim is immediately implied. If $g(D) < \var{D}$ then we must have $g(D) = g(D - x_i) + \Delta$ for some $x_i$. Applying the inductive hypothesis on $g(D - x_i)$ gives our desired claim.
\end{proof}

With this lemma in hand, the main idea for bounding accuracy is to condition on the size of $D'$, which we denote $k$, and give bounds separately for the cases when $k\leq n/2$ and $k \geq n/2$.
When $k$ is small we will just bound our error by $\var{D} - (n-k)\Delta$ and use the fact that $(n-k)\Delta$ is large.
When $k$ is large we will look at the loss in accuracy from $\var{D} - \var{D'}$ where we will bound this by iteratively applying the following lemma.

\begin{lemma}\label{lem:marginal_variance_bound}
	For any  $D = (x_1,...,x_n) \in \R^n$ and any $x_a \in D$, then 
	\[
	\var{D} - \var{D - x_a} \leq \frac{1}{n^2} \sum_{i=1}^n \left(x_a - x_i \right)^2. 
	\]
\end{lemma}

\begin{proof}
	By the definition of variance,
	\[
		\var{D} - \var{D - x_a} = \frac{1}{n^2} \sum_{i=1}^n\sum_{j=1}^n\frac{1}{2}\left(x_i - x_j\right)^2 - \frac{1}{(n-1)^2} \sum_{i\neq a}\sum_{j\neq a}\frac{1}{2}\left(x_i - x_j\right)^2.
	\]
	 This reduces to,
	 \[
	 \var{D} - \var{D - x_a} = 
	  \frac{1}{n^2} \sum_{i=1}^n\left(x_a - x_i\right)^2 - \frac{2n-1}{n^2(n-1)^2} \sum_{i\neq a}\sum_{j\neq a}\frac{1}{2}\left(x_i - x_j\right)^2,
	 \]
	 which gives our desired equality.
\end{proof}

Recall that we want to use this lemma to bound $\var{D} - \var{D'}$ where $D'$ is a subset of $D$ with size $k$. 
Suppose $D' = (x_1,...,x_k)$ and let $D_i = (x_1,...,x_i)$ for any $i$; we will use the fact that $\var{D} - \var{D'} = \sum_{i=k+1}^n \var{D_i} - \var{D_{i-1}}$.  The above Lemma \ref{lem:marginal_variance_bound} allows us to bound this sum, which will be the key step in our accuracy bounds.

\begin{lemma}\label{lem:variance_error}
	Given $D = (x_1,...,x_n) \in \R^n$ and $g$ that is the \sensitivity~of variance with parameter $\Delta$ and $g(\emptyset) = 0$, then
	\[
	\abs{\var{D} - g(D)} \leq \max\left\{\var{D} - \frac{n}{2}\Delta,0\right\} + \sum_{i=1}^n \max \left\{\sum_{j=1}^n \frac{4(x_i - x_j)^2}{n^2} - \Delta, 0\right\}.
	\]
\end{lemma}

\begin{proof}
	Note that Lemma~\ref{lem:variance_doesnt_go_above} implies $\abs{\var{D} - g(D)} = \var{D} - g(D)$.
	From Lemma~\ref{lem:simplification_of_variance_value} we know that $g(D) = \var{D'} + (n-k)\Delta$ for some $D' \subseteq D$ of size $k$, and we can rewrite $\var{D} - g(D) = \var{D} - \var{D'} - (n-k)\Delta$.
	If $k \leq n/2$, then 
	\[
	\var{D} - g(D) \leq \var{D} - \frac{n}{2}\Delta,
	\]
	because $(n-k) \geq n/2$ and $\var{D'} \geq 0$.
	
	If $k \geq n/2$, then for simplicity we will assume $D' = (x_1,...,x_k)$ and address this assumption later. We then let $D_i = (x_1,...,x_i)$ for any $i$ and use the fact that $\var{D} - \var{D'} = \sum_{i=k+1}^n \var{D_i} - \var{D_{i-1}}$. Lemma~\ref{lem:marginal_variance_bound} along with the fact that $k \geq n/2$ allows us to then bound this summation as
	\[
	\sum_{i=k+1}^n \var{D_i} - \var{D_{i-1}} \leq \frac{4}{n^2} \sum_{i=k+1}^n \sum_{j=1}^n (x_i - x_j)^2
	\]
	We can then use this to achieve (for $D' = (x_1,...,x_k)$)
	\[
	\var{D} - \var{D'} - (n-k)\Delta \leq \sum_{i=k+1}^n \left(\sum_{j=1}^n \frac{4(x_i - x_j)^2}{n^2} - \Delta\right)
	\]
	
	At this point we address the assumption that $D' = (x_1,...,x_k)$ by simply adding non-negative terms to the summation and ensuring that all of the entries in $D'$ are be included in this summation.
	This gives us,
	\[
	\var{D} - \var{D'} - (n-k)\Delta \leq \sum_{i=1}^n \max \left\{\sum_{j=1}^n \frac{4(x_i - x_j)^2}{n^2} - \Delta, 0\right\}.
	\]
	Adding both errors for $k \leq n/2$ and $k \geq n/2$ gives our desired bound.
\end{proof}

\subsection{Proof of Theorem~\ref{thm:variance}}

We now have all the necessary pieces for Theorem~\ref{thm:variance}, which we restate and prove here.

\variance*

\begin{proof}
	The fact that $g$ has sensitivity $\Delta$ follows from the fact that it is our \sensitivity~from Lemma~\ref{lem:variance_runtime_and_correctness}, 	
	and the guarantees of Lemma~\ref{lem:1d_sensitivity_guarantees}.
	The runtime also follows from Lemma~\ref{lem:variance_runtime_and_correctness}.
	We then achieve the error bounds from Lemma~\ref{lem:variance_error}.
	
\end{proof}

% !TEX root = main.tex

\section{Sensitivity preprocessing for personalized privacy guarantees}\label{s.personal}

%\rc{for this section, make sure we have DP and global sensitivity defined in prelims}

%\todo{Para:} Our sensitivity preprocessing can be used to give personalized privacy guarantees.

In this section, we introduce \emph{personalized differential privacy}, where each individual in a database may receive a different privacy parameter $\eps_i$.  We show that our \sensitivity~is naturally compatible with this privacy notion, and demonstrate the use of sensitivity-bounded functions for achieving personalized privacy guarantees, using the Laplace Mechanism and the Exponential Mechanism as illustrative examples.  The notion of personalized privacy has been previously applied to the design of markets for privacy.  We demonstrate the use of \sensitivity~for this application in Section \ref{s.markets}, and hope that our results may be useful tools for this well-studied problem in algorithmic economics.

\subsection{Personalized differential privacy}

We begin by defining \emph{personalized differential privacy}, which extends the standard definition of differential privacy (Definition \ref{def.dp}) to a setting where different individuals participating in the same computation may experience different, personalized privacy guarantees.  Similar definitions have also been used in previous work \cite{JYC15, ESS15, AGK17, LXJJ17}.  Recall from Section \ref{s.prelims} that two databases are neighboring if they differ in at most one entry.  We will say that two databases are \emph{$i$-neighbors} if they differ only in the $i$-th entry.

%\rc{$(\eps,0)$ or $(\eps,\delta)$? idk if it works with $\delta$ but I haven't thought about it}

\begin{definition}[Personalized differential privacy]
A mechanism $\cM: \mathcal{D} \rightarrow \mathcal{R}$ is \emph{$\{\epsilon_i\}$-personally differentially private} if for all $i$, for every pair of $i$-neighbors $D, D' \in \mathcal{D}$, and for every subset of possible outputs $\mathcal{S} \subseteq \mathcal{R}$,
\[ \Pr[\cM(D) \in \mathcal{S}] \leq \exp(\epsilon_i)\Pr[\cM(D') \in \mathcal{S}]. \]
%If all $\delta_i = 0$, we say that $\cM$ is {\em $\{\epsilon_i\}$-personally differentially private}.
\end{definition}

Note that any $\{\epsilon_i \}$-personally differentially private algorithm is also $(\max_i \eps_i)$-differentially private, since differential privacy provides a worst-case guarantee over all pairs of neighboring databases.

%\todo{Para:} compare to related work on personal privacy \cite{GR13}, local privacy, others?

%\todo{para:} relate to our results above. 

In this section, we show that personalized differential privacy can be achieved by combining our sensitivity preprocessing step with existing differentially private mechanisms.  An analyst can first apply our preprocessing step to get $g$ with desired individual sensitivity bounds, and then evaluate $g$ using a differentially private algorithm. The resulting $\{\epsilon_i\}$-personal differential privacy guarantees will depend on the chosen sensitivity parameters $\{\Delta_i\}$.  Since the function $g$ is independent of the database, the sensitivity preprocessing step does not leak any additional privacy.

%\todo{para:} Say something about accuracy. 

%\todo{define individual sensitivity, show how to achieve it, incorporate into existing models (eg Laplace, gaussian, Exp, ANT/Sparse)}

Individual sensitivity guarantees are critical for accurate analysis in this new privacy model.  Using only global sensitivity bounds $\Delta$, personally differentially private mechanisms add noise that scales with $\max_i \{\Delta/\eps_i \}$.  This alone cannot offer significant accuracy improvements because the noise must still scale inversely proportionally to the smallest $\eps_i$.   By utilizing individual sensitivity bounds, an analyst can tune each $\Delta_i$ to scale with $\eps_i$ to achieve overall accuracy improvements with personalized differential privacy.

%\rc{david can you write this para? purpose is to compare accuracy to other mechanisms.} They'll add noise that scales with max $\Delta/\eps_i$, where $\Delta$ is GS.  We shrink to $\Delta_i$, decreasing noise parameter but incurring some error in the process. Giving personalized privacy guarantees doesn't help accuracy bc still scales with worst eps. min $\eps_i$ determines everything.  We can tune $\Delta_i$ to scale with $\eps_i$, allowing accuracy improvements with personalized privacy.

We note that \emph{local differential privacy} \cite{KLN+08} also affords different privacy guarantees to different individuals in the same database, by perturbing each user's data locally before submitting it to the database.  Significantly stronger accuracy guarantees are possible in the presence of a trusted curator---which we assume in our model---because the analyst can leverage correlation of noise across individuals \cite{Ull18}.

A formal statement of the privacy and accuracy guarantees that arise from applying differentially private algorithms to sensitivity-bounded functions will depend on the exact algorithm used.  We illustrate this approach below applying it on two of the most foundational differentially private algorithms: the Laplace Mechanism and the Exponential Mechanism.  %We then show that three key properties of differential privacy continue to hold for personalized differential privacy: Composition, Post-processing, and Group privacy.  These properties together allow for modular private algorithm design using the three mechanisms presented here as building blocks for more advanced algorithms.

%Several lemmas showing how to do personalized privacy: laplace, gaussian, expo, etc.  

\subsubsection*{Laplace Mechanism}

The \emph{Laplace Mechanism} \cite{DMNS06} is perhaps the most fundamental of all differentially private algorithms.  It first evaluates a real-valued function $f$ on an input database $D$, and then perturbs the answer by adding Laplace noise scaled to the global sensitivity of $f$ divided by $\eps$.  The \emph{Laplace distribution} with scale $b$, denoted $\Lap(b)$, has probability density function:
\[ \Lap(x|b) = \frac{1}{2b} \exp\left(-\frac{|x|}{b}\right). \]
%We will write $\Lap(b)$ to denote both the Laplace distribution with scale $b$ and a random variable drawn from the the Laplace distribution with scale $b$.

\begin{definition}[Laplace Mechanism \cite{DMNS06}]\label{def.laplacemech}
Given any function $f : \cD \to \mathbb{R}$, the \emph{Laplace Mechanism} is defined as,
\[ \cM_L (D, f, \Delta f / \eps) = f(D) + Y, \]
where $Y$ is drawn from $\Lap(\Delta f / \eps)$.\footnote{We note that the standard definition of the Laplace Mechanism in \cite{DMNS06} takes $\eps$ as input instead of $\frac{\Delta f}{\eps}$.  We use the latter here for ease of notation when extending to personalized differential privacy.  This change does not affect the algorithm at all.}
\end{definition}

The Laplace Mechanism is $\eps$-differentially private \cite{DMNS06}.  We now show how to combine the Laplace Mechanism with our \sensitivity~to achieve personalized differential privacy guarantees.

%Privacy proofs: Show for each one that it can be make $\eps_i$ personalized private for desired input $\eps$ vector.  

\begin{proposition}\label{prop.laplace}
Let $g: \cD \to \R$ be a function with individual sensitivities $\{\Delta_i\}$.  For any $\{\eps_i\}$, the Laplace Mechanism $\cM_L \left(D, g, \max_j \{\Delta_j  / \eps_j \} \right)$ is $\{ \eps_i \}$-personally differentially private.
%The Laplace Mechanism applied to a function with individual sensitivities $\{\Delta_i\}$ is $\{\frac{\eps \Delta_i}{\max_j \Delta_j}\}$-personally differentially private.
\end{proposition}
\begin{proof}
Let $D, D' \in \cD$ be $i$-neighbors, let $g : \cD \to \mathbb{R}$ be a function with individual sensitivities $\{\Delta_i\}$, and let $r \in \bR$ be arbitrary.
\begin{align*}
\frac{\Pr[\cM_L(D, g,\max_j \{\Delta_j  / \eps_j \})=r]}{\Pr[\cM_L(D', g,\max_j \{\Delta_j  / \eps_j \})=r]} &= \frac{\exp\left( -\min_j\{\frac{\eps_j}{\Delta_j} \} |g(D)-r|\right)}{\exp\left( -\min_j\{\frac{\eps_j}{\Delta_j} \} |g(D')-r|\right)} \\
&= \exp \left( \min_j\{\frac{\eps_j}{\Delta_j} \}  \left( |g(D') - r| - |g(D) - r| \right) \right) \\
&\leq \exp \left( \min_j\{\frac{\eps_j}{\Delta_j} \} \left( |g(D) - g(D')| \right) \right) \\
&\leq \exp \left(\min_j\{\frac{\eps_j}{\Delta_j} \} \Delta_i \right) \\
%&\leq \exp \left(\frac{\eps_i}{\Delta_i} \Delta_i \right) \\
&\leq \exp(\eps_i)
\end{align*}
%The first inequality comes from the triangle inequality and the second inequality comes from the definition of sensitivity.
Then this version of the Laplace Mechanism run on a function with individual sensitivities $\{\Delta_i\}$ is $\{\eps_i \}$-personally differentially private.
%To achieve a desired $\{\eps_i\}$-personalized privacy guarantee for some function $f$ with arbitrary sensitivities, one could first find a $\{\Delta_i\}$-sensitivity preprocessing function $g$ with $\Delta_i = \frac{\eps_i}{\eps} \max_j \Delta_j$ for all $i$.  Since the function $g$ depends only on $f$ and not on the database $D$, the sensitivity preprocessing step does not leak any additional privacy.
\end{proof}

%Since the function $g$ is independent of the database, the sensitivity preprocessing step does not leak any additional privacy.

Proposition \ref{prop.laplace} shows that to achieve personalized privacy guarantees for a given function $f$, one can apply our \sensitivity~to produce Sensitivity-Bounded $g$, and then apply the Laplace Mechanism.  The accuracy guarantees of this procedure will depend on the worst-case ratio of $\Delta_i/\eps_i$, as well as global sensitivity of the original function $f$.  If one person $j$ requires significantly higher privacy protections than the rest of the population, the analyst can account for this by reducing $\Delta_j$.  This may greatly improve accuracy over the standard approach, which would require the analyst to add increased noise to the entire population.
%There is a fundamental tradeoff when choosing the $\{\Delta_i\}$ parameters of \sensitivity.  Smaller $\Delta_i$ mean less noise added to preserve privacy, but more error in the preprocessing approximation.  Further, if one 
 We address this challenge more concretely in Section \ref{s.markets}, using the application of market design for private data.

%\todo{accuracy} \rc{David can you do this too?} Meta-accuracy theorem showing that accuracy depends on worst-case ratio of $\Delta_i/\eps_i$.  Trade-off: smaller $\Delta_i$ means less noise added to preserve privacy, but more error in the pre-processing approximation.  Think about overall accuracy of estimator when setting $\Delta_i$s. e.g. everyone has $\Delta_i = 1$ and $\eps_i = 1/n$ to everyone except player $j$, who has $\eps_j = 2/n$.  Wouldn't drop everyone else's $\Delta_i$ to be 1/2 (for most accuracy definitions).  Cheater approach: can always drop sensitivity by scaling down function, but that doesn't increase accuracy bc gains in noise are also scaled down by same amount.

\subsubsection*{Exponential Mechanism}

The \emph{Exponential Mechanism} \cite{MT07} is a powerful private mechanism for answering non-numeric queries with an arbitrary range, such as selecting the best outcome from a set of alternatives.  The quality of an outcome is measured by a \emph{score function} $q \colon \cD \times \cR \to \mathbb{R}$, which relates each alternative to the underlying data through a real-valued score.  The global sensitivity of the score function is measured only with respect to the database argument; it can be arbitrarily sensitive in its range argument:
\[ \Delta q = \max_{r \in \cR} \max_{D,D' \; neighbors} | q(D,r) - q(D',r)|. \]
We define the individual sensitivity of a quality score analogously with respect to only its database argument:
\[ \Delta_i(q) = \max_{r \in \cR} \max_{D,D' \; i-neighbors} | q(D,r) - q(D',r)|. \]

The Exponential Mechanism samples an output from the range $\cR$ with probability exponentially weighted by score.  Outcomes with higher scores are exponentially more likely to be selected, thus ensuring both privacy and a high quality outcome.

\begin{definition}[Exponential Mechanism \cite{MT07}]\label{def.expmech}
Given a quality score $q : \cD \times \cR \to \mathbb{R}$, the \emph{Exponential Mechanism} is defined as:\footnote{As with the Laplace Mechanism, we define the Exponential Mechanism to take $\frac{\Delta q}{\eps}$ as input, instead of $\eps$.  This change is purely notational, and has no impact on the algorithm.}
\[ \cM_E(D, q, \Delta q /\eps) = \mbox{output } r \in \cR \mbox{ with probability proportional to } \exp\left(\frac{\eps q(D, r)}{2 \Delta q}\right).\]
\end{definition}

The Exponential Mechanism is $\eps$-differentially private \cite{MT07}.  We now show that when a score function has bounded individual sensitivity, the Exponential Mechanism is personally differentially private.

\begin{proposition}\label{prop.expo}
Let $q: \cD \times \cR \to \mathbb{R}$ be a score function with individual sensitivities $\{\Delta_i\}$.  For any $\{\eps_i\}$, the Exponential Mechanism $\cM_E \left(D, q, \max_j \{\Delta_j  / \eps_j \} \right)$ is $\{ \eps_i \}$-personally differentially private.
\end{proposition}
\begin{proof}
Let $D, D' \in \cD$ be $i$-neighbors, let $q$ be a score function with individual sensitivities $\{\Delta_i\}$, and let $r \in \cR$ be an arbitrary element of the output range.
\begin{align*}
&\frac{\Pr[\cM_E(D,q,\max_j \{\Delta_j  / \eps_j \})=r]}{\Pr[\cM_E(D',q,\max_j \{\Delta_j  / \eps_j \})=r]} \\
&\quad \quad = \frac{\left(\frac{\exp\left(\min_j \{\frac{\eps_j}{\Delta_j}\}q(D,r)/2 \right) }{\sum_{r' \in \cR} \exp\left(\min_j \{\frac{\eps_j}{\Delta_j}\} q(D,r')/2 \right) } \right)}{\left(\frac{\exp\left(\min_j \{\frac{\eps_j}{\Delta_j}\}q(D',r)/2 \right)}{\sum_{r' \in \cR} \exp\left(\min_j \{\frac{\eps_j}{\Delta_j}\} q(D',r')/2 \right) }\right) }\\
&\quad \quad = \left( \frac{ \exp\left(\min_j \{\frac{\eps_j}{\Delta_j}\} q(D,r)/2 \right)}{\exp\left(\min_j \{\frac{\eps_j}{\Delta_j}\} q(D',r)/2 \right)}\right) \cdot \left( \frac{\sum_{r' \in \cR} \exp\left(\min_j \{\frac{\eps_j}{\Delta_j}\}q(D',r')/2 \right) }{\sum_{r' \in \cR} \exp\left(\min_j \{\frac{\eps_j}{\Delta_j}\}q(D,r')/2 \right) }\right) \\
&\quad \quad = \exp \left( \min_j \{\frac{\eps_j}{\Delta_j}\} \left( q(D,r) - q(D',r) \right)/2 \right) \cdot \left( \frac{\sum_{r' \in \cR} \exp\left(\min_j \{\frac{\eps_j}{\Delta_j}\}q(D',r')/2 \right) }{\sum_{r' \in \cR} \exp\left(\min_j \{\frac{\eps_j}{\Delta_j}\}q(D,r')/2 \right) }\right) \\
&\quad \quad \leq \exp \left(\frac{1}{2} \min_j \{\frac{\eps_j}{\Delta_j}\} \Delta_i \right) \cdot \exp \left( \frac{1}{2}\min_j \{\frac{\eps_j}{\Delta_j}\} \Delta_i \right) \cdot \left( \frac{\sum_{r' \in \cR} \exp\left(\min_j \{\frac{\eps_j}{\Delta_j}\}q(D,r')/2 \right) }{\sum_{r' \in \cR} \exp\left(\min_j \{\frac{\eps_j}{\Delta_j}\} q(D,r')/2 \right) }\right) \\
&\quad \quad =\exp\left(\min_j \{\frac{\eps_j}{\Delta_j}\} \Delta_i \right) \\
&\quad \quad \leq \exp(\eps_i)
\end{align*}
\end{proof}

%\rc{accuracy thm commented out}
%Theorem \ref{thm.exputil} says that the probability of outputting a ``bad'' outcome decays exponentially quickly in the distance from the optimal output.
%
%\begin{theorem}[\cite{MT07}]\label{thm.exputil}
%Let $r \in \cR$ be the output of $\cM_E(D,q,\eps)$.  Then: 
%\[ \Pr \left[ q(D,r) \leq \max_{r' \in \cR} q(D, r') - \frac{2 \Delta q \left( \ln |\cR| + t \right)}{\eps} \right] \leq e^{-t}. \]
%Or equivalently:
%\[ \Pr \left[ q(D, r) \geq \max_{r' \in \cR} q(D, r') - \frac{2\Delta q \left( \ln(|\cR| / \beta)\right)}{\eps} \right] \geq 1-\beta. \]
%\end{theorem}

%\todo{para:} Exp is canonical $(\eps,0)$-mechanism so if we can do this then we can do them all. 

\begin{remark}\label{rem.expo}
The Exponential Mechanism is a canonical $\eps$-differentially private algorithm: every $\eps$-differentially private algorithm $\cM$ can be written as an instantiation of the Exponential Mechanism using quality score $q(D,r) = \ln ( \Pr[\cM(D)=r] )$ with global sensitivity $\Delta q = \eps$.  We can use this reduction to show that \emph{any} $\eps$-differentially private algorithm can be modified to give personal privacy guarantees using our \sensitivity.   First, re-write private mechanism $\cM$ as an Exponential Mechanism $\cM_E$, and then perform Sensitivity-Preprocessing on the quality score of $\cM_E$.  Proposition \ref{prop.expo} shows that the sensitivity-bounded version of $\cM_E$ will satisfy personalized differential privacy.
\end{remark}

%\subsubsection{Above Noisy Threshold}
%
%\begin{proposition}[Above Noisy Threshold]
%
%\end{proposition}
%\begin{proof}
%
%\end{proof}

%Privacy proofs: Show for each one that it can be make $\eps_i$ personalized private for desired input $\eps$ vector.  

%\subsubsection{Properties of Personal Differential Privacy}
%
%\begin{proposition}[Properties]
%
%\end{proposition}
%\begin{proof}
%
%\end{proof}

\subsection{Application: Markets for privacy}\label{s.markets}

One motivating application for wanting personalized privacy guarantees comes from algorithmic game theory and the study of market design for privacy.  This is a well-studied problem in the algorithmic economics community \cite{CCKMV13, NOS12, NST12, LR12, FL12, GR13, GLRS14, CLR+15, CIL15, WFA15, CPWV16}, and of practical importance as growing amounts of data are collected about individuals.  In a \emph{market for privacy}, a data analyst wishes to purchase and aggregate data from multiple strategic individuals.  These individuals may have privacy concerns, and will require compensation for their privacy loss from sharing data.  On the opposite side of the market, firms demand accurate estimates of population statistics, for uses such as market research or operational decision making.

The analyst must first purchase data from these strategic individuals, and then aggregate the collected data into an accurate estimate for firms.  Her goal is to perform this task while maximizing her own profits.  One of the tools at her disposal is differential privacy: by offering individuals formal privacy guarantees, their privacy costs from sharing data are diminished, and the analyst can provide smaller payments.  However, the noise from differential privacy may introduce additional error.

It is the analyst's task to determine the optimal privacy level for the market that balances these opposing effects.  Due to potentially heterogeneous privacy costs of the individuals, it may be optimal in terms of her profit for the analyst to provide different privacy guarantees to different individuals in the population.  She could then use our \sensitivity~to algorithmically provide the heterogeneous privacy levels demanded by the market.  We leave the challenge of modeling specifics of these markets as an open question to the algorithmic game theory community, and hope that our preprocessing tool and mechanisms for personalized privacy will open new avenues for designing markets for privacy.

%She may also need to use complicated pricing functions to elicit truthful data or to incentivize participation in her market.

%\todo{para} heterogeneous demand for privacy. privacy cost functions depend on privacy guarantee their given: $c_i(\eps) = c_i \eps$ or $c_i \eps^2$. cite stuff

%\todo{market maker para} market maker must decide what data/accuracy to buy from whom and at what price, aggregate in privacy-preserving way.  either minimize budget or maximize profit. need also to satisfy IR+IC.

%\todo{dont actually solve, just discuss}. When solving for this market's equilibrium, will result in different $\eps_i$ for different people.  Optimization of $\eps_i$ comes from MM's objective and market conditions. If constrained to single $\eps$, could be very bad for MM objective (budget/profit).  Accuracy for sale paper showed that optimum should depend on demand.  MM can compute optimal $\eps_i$ and implement using our preprocessing + any DP mech.

% !TEX root = main.tex

\section{Extension to 2-dimensions for $\ell_1$ sensitivity}\label{sec:higher_dimensions}

%\rc{moved this section to the end until we know where it belongs}

In this section we show that our \sensitivity~can be naturally extended to functions that map to 2-dimensional space where we consider the sensitivity in the $\ell_1$ distance metric.
While there is a natural extension of our \sensitivity~to higher dimensions, the primary difficulty will be ensuring that our greedy construction still yields a non-empty intersection of the constraints.
Interestingly, we show that this set of constraints
will give a non-empty intersection for 2 dimensions, and provide a counter-example for higher dimensions.
 
Recall that our \sensitivity~found a range $[\lowerbound(D),\upper(D)]$ where it could feasibly place $g(D)$, then choose the point in that segment closest to $f(D)$.
This range of feasible points came from intersecting each constraint $[g(D - x_i) - \Delta_i, g(D- x_i) + \Delta_i]$ induced by the neighbors of $D$ that are strictly smaller.
The \sensitivity~then chose the point in this intersection closest to $f(D)$. The key property needed by the algorithm was that this intersection was non-empty.

To prove this key property we took advantage of the data universe structure, which immediately yielded the fact that for any $x_i,x_j \in D$ we must have  $[g(D - x_i) - \Delta_i, g(D- x_i) + \Delta_i] \cap [g(D - x_j) - \Delta_j, g(D- x_j) + \Delta_j] \neq \emptyset$.
As a result, we had a finite set of line segments whose intersection was pair-wise non-empty, which immediately implies that the intersection of all line segments was non-empty.
We note that $[g(D - x_i) - \Delta_i, g(D- x_i) + \Delta_i]$ is the $\ell_1$ ball with radius $\Delta_i$ around $g(D - x_i)$ in one dimension.
For higher dimensions, the constraints will now be the $\ell_1$ ball with radius $\Delta_i$ around $g(D - x_i)$ in $d$ dimensions.
The structure of the data universe will still give that each of these $\ell_1$ balls has a non-empty pair-wise intersection.
However, this only implies that the intersection of all these $\ell_1$ balls is non-empty if we are in 2 dimensions. 
We first formally define the notion of an $\ell_1$ ball in higher dimensions.% as it will be critical to the rest of this section.

\begin{definition}[$\ell_1$ ball]\label{def:l1_ball}
	The $\ell_1$ ball around point $x^* \in \R^d$ with radius $\Delta$ is the set:
	\[
	(x^*,\Delta)^d_1 \defeq \{ x \in \R^d | \norm{x - x^*}_1 \leq \Delta \}.
	\]	
\end{definition}

To ensure that our choice of $g(D)$ does not violate the individual sensitivity parameter $\Delta_i$, we must place $g(D) \in (g(D - x_i),\Delta_i)^d_1$.
In the one-dimensional case, we had the same constraints, but they were simpler to handle because the $\ell_1$ ball is simply a line segment.
%We then chose the point among those constraints closest to $f(D)$, and we will equivalently 
We now define our \sensitivity~for two dimensions, which chooses the point that satisfies our constraints and is closest to $f(D)$, just as in the one-dimensional case.

\begin{definition}[2-dimensional Sensitivity-Preprocessing Function]
	Given a function $f: \mathcal{D} \rightarrow \R^2$ for any data universe such that for any $D \in \mathcal{D}$, all $D' \subset D$ are also in $\mathcal{D}$. For any non-negative individual sensitivity parameters $ \{\Delta_i\}$, we say that a function $g: \mathcal{D} \rightarrow \R^2$ is a \sensitivity~of $f$ with parameters $\{\Delta_i\}$ if $g(\emptyset) = f(\emptyset)$ and 
	\[g(D)= 
	 \text{closest point in } \cap_{x_i \in D} (g(D - x_i),\Delta_i)^2_1 \text{ to } f(D) \text{ in the $\ell_2$ metric}.
	\]
	If all $\Delta_i = \Delta$ for some non-negative $\Delta$, then we say that $g$ is a \sensitivity~of $f$ with parameter $\Delta$.
\end{definition}

Our primary goal of this section will then be to prove the following theorem that is equivalent to Theorem~\ref{thm:main_1d} but works for 2-dimensions. We will also point out the key spot within the proof where it breaks for dimensions greater than 2.

\begin{theorem}\label{thm:main_higher_dim}
	Given $T(n)$-time query access to an arbitrary $f:\mathcal{D} \rightarrow \mathbb{R}^2$, and sensitivity parameters $\{\Delta_i\}$, we provide $O((T(n) + n)2^{n})$ time access to \sensitivity~$g:\mathcal{D} \rightarrow \mathbb{R}^2$ such that $\Delta_i(g) \leq \newSens_i$.  Further, for any database $D = (x_1,\ldots,x_n)$,	
	\[ 
	\norm{f(D) - g(D)}_1 \leq \max_{\sigma \in \sigma_D}\sum_{i=1}^{|D|} \max \{ \norm{f(D_{\sigma(<i)} + x_{\sigma(i)}) - f(D_{\sigma(<i)})}_1 - \newSens_{\sigma(i)}, 0 \},
	\]
	where $\sigma_D$ is the set of all permutations on $[n]$, and $D_{\sigma(<i)} = (x_{\sigma(1)},...,x_{\sigma(i-1)})$ is the subset of $D$ that includes all individual data in the permutation before the $i$th entry.
\end{theorem}

As before, we will break the proof of this theorem into two parts. It immediately follows from construction that our 2-dimensional \sensitivity~will have the appropriate individual sensitivity parameters, but only if the function is well-defined. To this end, we first show in Section~\ref{subsec:2D_sensitivity} that the intersection of the $\ell_1$ balls is always non-empty if each pair-wise intersection is non-empty. Then in Section~\ref{subsec:2D_error} we give the analogous error guarantees where the proof will just follow equivalently to the one-dimensional case.

\subsection{Correctness of Sensitivity-Preprocessing Function}\label{subsec:2D_sensitivity}

In this section we show that for our 2-dimensional \sensitivity, it is always the case that $g(D)$ is defined.  This is equivalent to showing:
\[
\bigcap_{x_i \in D} (g(D - x_i),\Delta_i)^2_1 \neq \emptyset.
\]

We will first take advantage of the structure of data universes to show that the pair-wise intersection is always non-empty. Then we will use the fact that pair-wise intersection of $\ell_1$ balls in 2-dimensions implies that the intersection of all $\ell_1$ balls is non-empty. Intuitively, this is because $\ell_1$ balls in 2-dimensions are simply rotated squares. Further, we will show that this is exactly the step that breaks the algorithm for higher dimensions. 

\begin{lemma}\label{lem:pairwise_intersect_highD}
	Given any $f: \mathcal{D} \rightarrow \R^2$ and desired sensitivity parameters $\{\Delta_i\}$, let $g: \mathcal{D} \rightarrow \R^2$ be the \sensitivity~with parameters $\{\Delta_i\}$. For any $D \in \mathcal{D}$ with at least two entries, assume that $g(D')$ is defined for any $D' \subset D$. Then for any $x_i,x_j \in D$,
	\[
	(g(D - x_i),\Delta_i)^2_1 \cap (g(D - x_j),\Delta_j)^2_1 \neq \emptyset.
	\]
\end{lemma}

Note that we have not yet proven that $g(D)$ is defined on all databases, so we will need to first assume that it is on all subsets of $D$.  Our proof of this fact will be done inductively.

\begin{proof}
	We use the fact that $D$ has at least two entries and consider the database $D - x_i - x_j$.
	Due to our assumption that $g(D')$ is defined on all $D' \subset D$, it follows from our construction of $g$ that
	\[
	\norm{g(D - x_i) - g(D - x_i - x_j)}_1 \leq \newSens_j,
	\]	 
	and 
	\[
	\norm{g(D - x_j) - g(D - x_i - x_j)}_1 \leq \newSens_i,
	\]
	Applying triangle inequality gives,
	\[
	\norm{g(D - x_i) - g(D - x_j)}_1 \leq \newSens_i + \newSens_j,
	\]
	which implies our claim by the definition of $\ell_1$ balls.
\end{proof}

With this pair-wise intersection property, it now remains to be shown that this implies the intersection of all $\ell_1$ balls is non-empty.
For this we prove a general fact about the intersection of $\ell_1$ balls in 2-dimensions.

\begin{lemma}\label{lem:pairwise_implies_all_intersection}
	Consider any set of points $y_1,...,y_n \in \R^2$, where we let $(y_i)_1$ and $(y_i)_2$ denote the respective coordinates of $y_i$. Consider any set of non-negative $\Delta_1,...,\Delta_n$. If for any $y_i,y_j$,
	\[
	(y_i,\Delta_i)^2_1 \cap (y_j,\Delta_j)^2_1 \neq \emptyset,
	\]
	then,
	\[
	\bigcap_{i=1}^n (y_i,\Delta_i)^2_1 \neq \emptyset.
	\]
\end{lemma}

	Our proof will first rewrite each $\ell_1$ ball as a set of 4 linear inequalities. From this interpretation we will then use two critical facts. First, each inequality has a corresponding parallel inequality in any other $\ell_1$ ball. Second, removing any one of these constraints gives an unbounded polytope.

\begin{proof}
	Further examination of Definition~\ref{def:l1_ball} shows that
	\[
	(y_i,\Delta_i)^2_1 \defeq \{ x \in \R^2 | \norm{x - y_i}_1 \leq \Delta_i \} = \{ x \in \R^2 | \abs{(x)_1 - (y_i)_1} + \abs{(x)_2 - (y_i)_2} \leq \Delta_i \}.
	\]
	 We then use a known trick of converting absolute values into linear inequalities where $\abs{x} \leq k$ becomes $x \leq k$ and $-x \leq k$.
	\begin{align*}
	(y_i,\Delta_i)^2_1 = & \{ x \in \R^2 | (x)_1  + (x)_2  \leq (y_i)_1 + (y_i)_2 + \Delta_i \}\\
	& \cap \{ x \in \R^2 | - (x)_1  - (x)_2  \leq - (y_i)_1 - (y_i)_2 + \Delta_i \} \\
	& \cap \{ x \in \R^2 | (x)_1  - (x)_2  \leq (y_i)_1 - (y_i)_2 + \Delta_i \} \\
	& \cap \{ x \in \R^2 | - (x)_1  + (x)_2  \leq - (y_i)_1 + (y_i)_2 + \Delta_i \}
	\end{align*}
	
	At this point we note that each of the balls have parallel inequalities, so we can use the following fact:
	\begin{multline*}
	\{ x \in \R^2 | (x)_1  + (x)_2  \leq (y_i)_1 + (y_i)_2 + \Delta_i \} \cap \{ x \in \R^2 | (x)_1  + (x)_2  \leq (y_j)_1 + (y_j)_2 + \Delta_j \} \\
	= \left\{ x \in \R^2 | (x)_1  + (x)_2  \leq \min\{ (y_i)_1 + (y_i)_2 + \Delta_i, (y_j)_1 + (y_j)_2 + \Delta_j \} \right\}.
	\end{multline*}
	We apply this fact to the full intersection and obtain,
	\begin{align*}
	\bigcap_{i=1}^n (y_i,\Delta_i)^2_1 = & \{ x \in \R^2 | (x)_1  + (x)_2  \leq \min_{i \in[n]} \{(y_i)_1 + (y_i)_2 + \Delta_i\} \}\\
	& \cap \{ x \in \R^2 | - (x)_1  - (x)_2  \leq \min_{i \in[n]} \{ - (y_i)_1 - (y_i)_2 + \Delta_i \} \} \\
	& \cap \{ x \in \R^2 | (x)_1  - (x)_2  \leq \min_{i \in[n]} \{(y_i)_1 - (y_i)_2 + \Delta_i\} \} \\
	& \cap \{ x \in \R^2 | - (x)_1  + (x)_2  \leq \min_{i \in[n]} \{- (y_i)_1 + (y_i)_2 + \Delta_i\} \}.
	\end{align*}
	
	Intuitively, if this intersection exists, it must be a rectangle that is rotated 45 degrees. 
	To more easily see this fact, we now multiply the second and fourth constraint by -1.
	\begin{align*}
	\bigcap_{i=1}^n (y_i,\Delta_i)^2_1 = & \{ x \in \R^2 | (x)_1  + (x)_2  \leq \min_{i \in[n]} \{(y_i)_1 + (y_i)_2 + \Delta_i\} \}\\
	& \cap \{ x \in \R^2 | (x)_1  + (x)_2  \geq \min_{i \in[n]} \{(y_i)_1 + (y_i)_2 - \Delta_i \} \} \\
	& \cap \{ x \in \R^2 | (x)_1  - (x)_2  \leq \min_{i \in[n]} \{(y_i)_1 - (y_i)_2 + \Delta_i\} \} \\
	& \cap \{ x \in \R^2 | (x)_1  - (x)_2  \geq \min_{i \in[n]} \{(y_i)_1 - (y_i)_2 - \Delta_i\} \}
	\end{align*}
	
	With this interpretation it is straightforward to see that $\bigcap_{i=1}^n (y_i,\Delta_i)^2_1 = \emptyset$ if and only if
	\[
	\min_{i \in[n]} \{(y_i)_1 + (y_i)_2 + \Delta_i\} < \min_{i \in[n]} \{(y_i)_1 + (y_i)_2 - \Delta_i \},
	\]
	or
	\[
	\min_{i \in[n]} \{(y_i)_1 - (y_i)_2 + \Delta_i\} < \min_{i \in[n]} \{(y_i)_1 - (y_i)_2 - \Delta_i\}.
	\]
	
	Let $k$ be the index that minimizes $(y_i)_1 + (y_i)_2 + \Delta_i$ and let $l$ be the index that minimizes $(y_i)_1 + (y_i)_2 - \Delta_i$. If 
	\[
	(y_k)_1 + (y_k)_2 + \Delta_k < (y_l)_1 + (y_l)_2 - \Delta_l,
	\]
	then we must have,
	\[
	\Delta_k + \Delta_l < (y_l)_1  - (y_k)_1 + (y_l)_2 - (y_k)_2  \leq \abs{(y_l)_1  - (y_k)_1} + \abs{(y_l)_2 - (y_k)_2},
	\]
	which contradicts our assumption that $(y_k,\Delta_k)^2_1 \cap (y_l,\Delta_l)^2_1 \neq \emptyset$. This follows identically for the second inequality, so therefore neither of them can hold and the intersection must be non-empty.
\end{proof}

We now remark that Theorem \ref{thm:main_higher_dim} cannot be extended to higher dimensions or to $\ell_p$ norms.

\begin{remark}
	For extending to dimensions greater than two, the proof breaks down at Lemma~\ref{lem:pairwise_implies_all_intersection}. Intuitively, we can still interpret each $\ell_1$ ball as a set of linear inequalities, however it no longer has the critical property that removing one of the constraints creates an unbounded polytope.
	More specifically, consider the following counter-example for 3 dimensions:
	
	Let $A = \{(1,1,-1),(1,-1,1),(-1,1,1),(1,-1,-1),(-1,1,-1),(-1,-1,1)\}$ and set $\Delta = 3$. It is not difficult to see that taking the intersection of $\Delta$-radius $\ell_1$ balls around each point in $A$ will only contain the origin $(0,0,0)$. We then consider adding the point $(3/2,3/2,3/2)$, and it is straightforward to verify that the $\ell_1$ distance between this point and any point in $A$ is at most $11/2 < 2\Delta$. However, the origin is not within the $\ell_1$ ball around $(3/2,3/2,3/2)$. Therefore, if we consider the set of $\ell_1$ balls of radius $\Delta$ around the points in $A \cup (3/2,3/2,3/2)$, then each pair of $\ell_1$ balls will intersect, but the full intersection  will be empty, giving our counter-example.
	
	%\todo{find the scratch paper with my counter-example for 3-dimensions}
	
\end{remark}

\begin{remark}
	Even in 2-dimensions, we cannot have Lemma~\ref{lem:pairwise_implies_all_intersection} for the $\ell_p$ ball with $p \in (1,\infty)$ due to the curvature of each ball. For instance, consider the $\ell_2$ ball with radius $1$ for the points $(-1,0),(1,0),(0,\sqrt{3})$. Each of pair of these points is exactly distance $2$ apart in the $\ell_2$ metric, so their $\ell_2$ balls of radius 1 each pairwise intersect. However it is easy to see that the intersection of all three is empty.
	
	We can similarly extend this counter-example to other $\ell_p$ balls using the fact that there must be some curvature of the $\ell_p$ ball, and the midpoint between any two points in the $\ell_p$ metric is unique if $p \in (1,\infty)$.

	%\rc{wanna say anything else about this?}
\end{remark}

With these lemmas, we are now able to show that our 2-dimensional \sensitivity~must always be defined.

\begin{lemma}\label{lem:g_is_defined_in_2D}
	Given any $f: \mathcal{D} \rightarrow \R^2$ with sensitivity parameters $\{\Delta_i\}$, let $g: \mathcal{D} \rightarrow \R^2$ be the \sensitivity~with parameters $\{\Delta_i\}$. Then for any $D \in \mathcal{D}$,
	\[
	\bigcap_{x_i \in D} (g(D - x_i),\Delta_i)^2_1 \neq \emptyset.
	\]
\end{lemma}

\begin{proof}
	We will prove this fact inductively, and note that it is immediately true when $D$ only has one entry.
	
	We then consider an arbitrary database $D$ and assume that it is true for all $D' \subset D$. With this inductive claim we can apply Lemma~\ref{lem:pairwise_intersect_highD} to get that all of the $\ell_1$ balls have non-empty pairwise intersection. Our desired result then immediately follows from applying Lemma~\ref{lem:pairwise_implies_all_intersection}.
\end{proof}

%We always start with a bounded polytope, so at each iteration when the polytope is feasible it must also be bounded (it is only shrinking). Therefore, take any point $x$ in the polytope, and suppose we remove the constraint associated with set $S \subseteq [n]$.

\subsection{Error bounds for the 2-dimensional extension}\label{subsec:2D_error}

The following lemma gives the desired error bounds on the 2-dimensional \sensitivity.

\begin{lemma}\label{lem:2D_error_bounds}
	Given any $f: \mathcal{D} \rightarrow \R^2$ and desired sensitivity parameters $\{\Delta_i\}$, let $g: \mathcal{D} \rightarrow \R^2$ be the \sensitivity~with parameters $\{\Delta_i\}$. Then for any $D \in \mathcal{D}$,
	\[ 
	\norm{f(D) - g(D)}_1 \leq \max_{\sigma \in \sigma_D}\sum_{i=1}^{|D|} \max \{ \norm{f(D_{\sigma(<i)} + x_{\sigma(i)}) - f(D_{\sigma(<i)})}_1 - \newSens_{\sigma(i)}, 0 \},
	\]
	where $\sigma_D$ is the set of all permutations of the set $[n]$, and let $D_{\sigma(<i)} = (x_{\sigma(1)},...,x_{\sigma(i-1)})$ be the subset of $D$ that includes all individual data in the permutation before the $i$th entry.	
\end{lemma}

%\begin{proof}
	The proof of Lemma \ref{lem:2D_error_bounds} follows identically to the proof of Lemma~\ref{lem:1D_error_bounds} where by replacing any instance of absolute value with the 1-norm.
%\end{proof}

We are finally ready to complete the proof of our main theorem for two dimensions.

\begin{proof}[Proof of Theorem~\ref{thm:main_higher_dim}]
	The individual sensitivity guarantees follow from the construction of $g$ and Lemma~\ref{lem:g_is_defined_in_2D}. The error bounds are given by Lemma~\ref{lem:2D_error_bounds}. It then remains to prove the running time. For each subset of $D$ we need to query $f$ which takes $T(n)$ time by assumption.
	We note that within the proof of Lemma~\ref{lem:pairwise_implies_all_intersection} we gave a construction for obtaining the intersection of $n$ different $\ell_1$ balls which could clearly be done in $O(n)$ time. Finding the closest point to $f(D)$ then takes $O(1)$ time for the polytope defined by four inequalities. Therefore, the running time is $T(n) + O(n)$ for each of the $2^n$ subsets, which implies the desired running time. 
\end{proof}

%\david{prove the 2-approximation if there is time. i highly doubt there will be, and i think its a non-trivial extension. pareto optimality might be easier}

% !TEX root = main.tex

\section{Future Directions}\label{s.future}

%\rc{putting it here for now}

We are especially interested in efficiently implementing our framework for more complicated and, in particular, higher-dimensional functions such as linear regression. 
We believe that leveraging the simple recursive construction of our algorithm along with non-trivial structural properties of these more difficult functions can allow for efficient and accurate implementation. 
We are particularly optimistic because all of our proofs in this work were from first principles, suggesting that we may be able to obtain further results from this framework by using more sophisticated tools.

While our construction did not generalize to any dimension under the $\ell_1$ sensitivity metric, we note that this was in the most general setting.
If the class of functions we consider is significantly restricted, then we believe the natural extension could both work and be efficiently implementable.
Furthermore, we have not yet investigated variants of our algorithm that might work better under stronger assumptions or combining our construction with other frameworks for handling worst-case sensitivity.

We also believe that our construction opens up several intriguing directions with respect to personalized differential privacy and its application in markets for privacy. 
Our construction allows for tailoring individual sensitivity, but this presents a natural trade-off between choosing small individual sensitivity parameters and the error incurred by our preprocessing step.
For specific functions, this may yield interesting optimization problems that can also be considered in the context of markets for privacy.

\section*{Acknowledgements}

We thank Richard Peng, Aaron Roth, and Jamie Morgenstern for useful feedback and discussions.  We also thank the participants of the Banff International Research Station Workshop on Mathematical Foundations of Data Privacy---especially Adam Smith---for comments and discussion on an earlier draft of this paper.

\bibliographystyle{alpha}
\bibliography{ref}

\appendix

% !TEX root = main.tex

\section{Omitted Proofs}

In this appendix we provide proofs that were omitted from Sections~\ref{sec:efficient_examples} and~\ref{sec:variance_algo}.

\subsection{Proof of Lemma~\ref{cor:bound_on_sum_of_absolute_values}}

	%\todo{Clean this up if there is time}

\begin{proof}[Proof of Lemma~\ref{cor:bound_on_sum_of_absolute_values}]
	We start by decomposing the RHS:
	\[
	\frac{1}{n} \sum_{i=k}^n \abs{x_i - \mu_D} = \frac{1}{n} \sum_{i=k}^n \abs{x_i - \frac{x_1 + \cdots + x_n}{n}} = \frac{1}{n^2} \sum_{i=k}^n \abs{\sum_{j=1}^n (x_i - x_j)}.
	\]
	It now suffices to show,
	\[
	\sum_{i=k}^n \sum_{j=1}^i \frac{1}{3} \abs{x_i - x_j} \leq 
	\sum_{i=k}^n \abs{\sum_{j=1}^n (x_i - x_j)}.
	\]
	
	We further examine the RHS and use our assumption that $x_1 \leq \cdots \leq x_n$, which implies that for each $i$,
	\[
	\abs{\sum_{j=1}^n (x_i - x_j)} = \max \{ \sum_{j=1}^i \abs{x_i - x_j} - \sum_{j=i}^n \abs{x_i - x_j}, \sum_{j=i}^n \abs{x_i - x_j} - \sum_{j=1}^i \abs{x_i - x_j}\}.
	\]
	
	The idea will then be that because we have an ordering on $x_1,...,x_n$, there will be a transition index. In particular, there is some $l \in [n-1]$ such that 
	\[
	\abs{\sum_{j=1}^n (x_i - x_j)} = \sum_{j=1}^i \abs{x_i - x_j} - \sum_{j=i}^n \abs{x_i - x_j}
	\]
	for all $i > l$, and
	\[
	\abs{\sum_{j=1}^n (x_i - x_j)} = \sum_{j=i}^n \abs{x_i - x_j} - \sum_{j=1}^i \abs{x_i - x_j} 
	\] 
	for all $i \leq l$.	If we then have $l \leq k$, then,
	\[
	\sum_{i=k}^n \abs{\sum_{j=1}^n (x_i - x_j)} = \sum_{i=k}^n \left(  \sum_{j=1}^i \abs{x_i - x_j} - \sum_{j=i}^n \abs{x_i - x_j}  \right).
	\] 
By cancellation, we get,
	\[
	\sum_{i=k}^n \abs{\sum_{j=1}^n (x_i - x_j)} = 
	\sum_{i=k}^n\sum_{j=1}^k \abs{x_i - x_j}.
	\]
	
	Applying Fact~\ref{fact:helper_for_fact_about_absolute_value_sum} gives,
	\[
	2 \sum_{i=k}^n\sum_{j=1}^k \abs{x_i - x_j} \geq \sum_{i=k}^n\sum_{j=1}^n \abs{x_i - x_j}
	\geq \sum_{i=k}^n\sum_{j=1}^i \abs{x_i - x_j},
	\]
	as desired. If $l > k$, then,
	\begin{multline*}
	\sum_{i=k}^n \abs{\sum_{j=1}^n (x_i - x_j)} =
	\\
	\sum_{i=k}^l \left( \sum_{j=i}^n \abs{x_i - x_j} - \sum_{j=1}^i \abs{x_i - x_j} \right) + \sum_{i=l+1}^n \left( \sum_{j=1}^i \abs{x_i - x_j} - \sum_{j=i}^n \abs{x_i - x_j} \right).
	\end{multline*}
We will simply lower bound the first term in the sum by 0, and note that Fact~\ref{fact:helper_for_fact_about_absolute_value_sum} (stated below) implies,
	\[
	\sum_{i=k}^l\sum_{j=l}^n \abs{x_i - x_j} \geq \sum_{i=k}^l \sum_{j=1}^i \abs{x_i - x_j}.
	\]
	
Furthermore, by cancellation, we get,	
	\[ 
	\sum_{i=k}^n \abs{\sum_{j=1}^n (x_i - x_j)} \geq \sum_{i=l+1}^n\sum_{j=1}^{l+1} \abs{x_i - x_j}.
	\]
	We then use the fact that,
	\[
	\sum_{i=l+1}^n\sum_{j=1}^{l+1} \abs{x_i - x_j} \geq \sum_{i=k}^l \sum_{j=1}^i \abs{x_i - x_j},
	\]
	and 
	\[
	\sum_{i=l+1}^n\sum_{j=1}^{l+1} \abs{x_i - x_j} \geq \sum_{i=l+1}^n \sum_{j=l}^i \abs{x_i - x+j} ,
	\]
	from Fact~\ref{fact:helper_for_fact_about_absolute_value_sum}, to obtain:
	\[
	3 \sum_{i=l+1}^n\sum_{j=1}^{l+1} \abs{x_i - x_j} \geq \sum_{i=k}^n \sum_{j=1}^i \abs{x_i - x_j},
	\]
	as desired.	
\end{proof}

\begin{fact}\label{fact:helper_for_fact_about_absolute_value_sum}
	For any ordered values $x_1 \leq \cdots \leq x_n$, and any $k \in [n-1]$ such that,
	\[
	\sum_{j=1}^k \abs{x_k - x_j} \leq \sum_{j= k +1}^n \abs{x_k - x_j},
	\]
	then for any $i \leq k$, we must have,
	\[
	\sum_{j=1}^k \abs{x_i - x_j} \leq \sum_{j= k +1}^n \abs{x_i - x_j}.
	\]
\end{fact}

\begin{proof}
	This follows from the fact that $x_1 \leq \cdots \leq x_k$, so $\sum_{j=1}^k \abs{x_i - x_j} \leq \sum_{j=1}^k \abs{x_k - x_j}$ and $\sum_{j= k +1}^n \abs{x_k - x_j} \leq \sum_{j= k +1}^n \abs{x_i - x_j}$.
\end{proof}

\subsection{Omitted proofs from Section~\ref{sec:variance_algo}}

In this section we prove some important facts about variance that were necessary for obtaining an efficient algorithm for variance.
We first show the intuitive fact that if we want to decrease the variance most, we should remove the maximum or minimum value.

\begin{proof}[Proof of Fact~\ref{fact:variance_minimized_with_min_max_tossed}]
	It suffices to show that $\var{D - x_1} \leq \var{D - x_i}$ if $x_i \leq \mu(D)$ and that $\var{D - x_n} \leq \var{D - x_i}$ if $x_i \geq \mu(D)$. We will show the first, and the second follows equivalently.
	
	We again use the definition of variance stated as,
	\[
	\var{x_1,...,x_n} = \frac{1}{n^2}\sum_{i=1}^n\sum_{j>i}^n\left(x_i-x_j\right)^2,
	\]
	and by cancellation we see that showing $\var{D - x_1} \leq \var{D - x_i}$ is equivalent to,
	\[
	\frac{1}{n^2} \sum_{j \neq 1} (x_i - x_j)^2 \leq \frac{1}{n^2} \sum_{j \neq i} (x_1 - x_j)^2,
	\]
	which is also equivalent to showing,
	\[
	\sum_{j \neq 1,i} (x_i - x_j)^2 \leq \sum_{j \neq 1,i} (x_1 - x_j)^2.
	\]
The proof then follows from the fact that this is a sum of least squares minimization for the vector $x_2,...,x_{i-1},x_{i+1},...,x_n$, where we know that $x_1 \leq x_i$ and $\mu({D - x_1 - x_i}) \geq \mu(D)$ because $x_1,x_i \leq \mu(D)$.
\end{proof}

We now prove Fact~\ref{fact:bound_marginal_variance} using the simple helper fact that if we add a data point and want to minimize the variance, then the added data point should the mean of the remaining points.

\begin{fact}\label{fact:mean_minimizes_variance}
	Given any set $x_1,...,x_n \in \mathbb{R}$ with mean $\mu$, then
	\[
	\arg\min_y\var{y,x_1,...,x_n} = \mu.
	\]
\end{fact}

\begin{proof}
	By definition,
	\[
	\var{x_1,...,x_n} = \frac{1}{n^2}\sum_{i=1}^n\sum_{j>i}^n\left(x_i-x_j\right)^2.
	\]
	The problem we are considering fixes $x_1,...,x_n$ and minimizes the variable $y$, so each term in the summation that does not include $y$ can be ignored, and our minimization problem then reduces to,
	\[
	\arg\min_y\var{y,x_1,...,x_n} = \arg\min_y \sum_{i=1}^n(y-x_i)^2,
	\]
	which is minimized when $y = \mu$.
\end{proof}

We use this fact to lower bound the variance from adding one additional variable, and complete our proof of Fact~\ref{fact:bound_marginal_variance}.

\begin{proof}[Proof of Fact~\ref{fact:bound_marginal_variance}]
	We first upper bound the variance of all variables:
	\[
	\min_y \var{x_1,...,x_{n-1},y} \leq \var{x_1,...,x_n}.
	\]
	Fact~\ref{fact:mean_minimizes_variance} implies that,
	\[
	\min_y \var{x_1,...,x_{n-1},y} = \var{x_1,...,x_{n-1},\mu_{[1:n-1]}},
	\]
	where $\mu_{[1:n-1]}$ is the mean of $x_1,...,x_{n-1}$. We then apply the definition of variance to get,
	\[
	\var{x_1,...,x_{n-1},\mu_{[1:n-1]}} = \frac{1}{n}\sum_{i=1}^{n-1}(x_i - \mu_{[1:n-1]})^2.
	\] 
	
	Together, this implies that,
	\[
	\min_y \var{x_1,...,x_{n-1},y} = \frac{n-1}{n} \var{x_1,...,x_{n-1}},
	\]
	which gives our desired result.
\end{proof}

Finally, we also needed the following fact to reduce our running time to $O(n^2)$ for implementation of variance. 

\begin{proof}[Proof of Fact~\ref{fact:variance_in_constant_time}]
	We utilize the definition of variance as,
	\[
	\var{x_1,...,x_n} = \frac{1}{n^2}\sum_{i=1}^n\sum_{j=1}^n\frac{1}{2}\left(x_i-x_j\right)^2.
	\]
After cancellation from the scalars we have,
	\begin{multline*}
	\left(\frac{n-1}{n}\right)^2\var{D - x_a} + \left(\frac{n-1}{n}\right)^2\var{D - x_b} - \left(\frac{n-2}{n}\right)^2\var{D - x_a - x_b}
	+ \frac{1}{n^2}(x_a - x_b)^2 \\
	= \frac{1}{n^2}\sum_{i\neq a}\sum_{j\neq a}\frac{1}{2}\left(x_i-x_j\right)^2
	+ \frac{1}{n^2}\sum_{i\neq b}\sum_{j\neq b}\frac{1}{2}\left(x_i-x_j\right)^2
	- \frac{1}{n^2}\sum_{i\neq a,b}\sum_{j\neq a,b}\frac{1}{2}\left(x_i-x_j\right)^2
	+ \frac{1}{n^2}(x_a - x_b)^2.
	\end{multline*}
	
	Separating out within the summation gives,
	\[
	\frac{1}{n^2}\sum_{i\neq a,b}\sum_{j\neq a,b}\frac{1}{2}\left(x_i-x_j\right)^2 + \frac{1}{n^2}\sum_{i\neq a}\left(x_b-x_i\right)^2
	+ \frac{1}{n^2}\sum_{i\neq b}\left(x_a-x_i\right)^2
	+ \frac{1}{n^2}(x_a - x_b)^2,
	\]
	which is equivalent to $\var{D}$ as desired.	
\end{proof}
	
\end{document}